\documentclass[10pt]{article}

\usepackage[top=2cm, bottom=2cm, left=2cm, right=2cm]{geometry}
\usepackage{amsmath,amssymb,amsfonts,amsthm,MnSymbol}
\usepackage{graphics,subfigure,color} 
\usepackage{latexsym}
\usepackage[dvips]{graphicx}
\usepackage{epsfig}
\usepackage{hyperref}
\usepackage{chngcntr}
\usepackage{psfrag}
\usepackage{multirow}
\usepackage{braket}
\usepackage{float}
\usepackage{enumerate}

\newcommand{\cB}{{\cal{B}}}
\newcommand{\cC}{{\cal{C}}}
\newcommand{\cD}{{\cal{D}}}
\newcommand{\cE}{{\cal{E}}}

\newcommand{\cH}{{\cal{H}}}

\newcommand{\cM}{{\cal{M}}}

\newcommand{\cS}{{\cal S}}
\newcommand{\cT}{{\cal{T}}}
\newcommand{\cV}{{\cal{V}}}

\newcommand{\mbT}{\mathbb{T}}

\DeclareMathOperator{\Tr}{Tr}

\newtheorem{definition}{Definition}

\newtheorem{lemma}{Lemma}
\newtheorem{theorem}{Theorem}

\theoremstyle{definition}

\allowdisplaybreaks[4]

\begin{document}

\title{Phase Transition in Tensor Models}
 
\author{
 \ Thibault Delepouve\footnote{delepouve@cpht.polytechnique.fr, Laboratoire de Physique Th\'eorique, CNRS UMR 8627, Universit\'e Paris Sud, 91405 Orsay Cedex, France
 and Centre de Physique Th\'eorique, CNRS UMR 7644, \'Ecole Polytechnique, 91128 Palaiseau Cedex, France.},
 \ Razvan Gurau\footnote{rgurau@cpht.polytechnique.fr, Centre de Physique Th\'eorique, CNRS UMR 7644, \'Ecole Polytechnique, 91128 Palaiseau Cedex, France
 and Perimeter Institute for Theoretical Physics, 31 Caroline St. N, N2L 2Y5, Waterloo, ON, Canada.}
}

\maketitle

\begin{abstract}

Generalizing matrix models, tensor models generate dynamical triangulations in any dimension and support a $1/N$ expansion.
Using the intermediate field representation we explicitly rewrite a quartic tensor model as a field theory for a fluctuation field 
around a vacuum state corresponding to the resummation of the entire leading order in $1/N$ (a resummation of the melonic family).
We then prove that the critical regime in which the continuum limit in the sense of dynamical triangulations is reached is precisely
a phase transition in the field theory sense for the fluctuation field.
\end{abstract}

\section{Introduction}

Matrix models \cite{DiFrancesco:1993nw} generate dynamical triangulations and dynamical tessellations in two dimensions. 
Tensor models \cite{ambj3dqg,sasa1,color,review,Tanasa:2011ur} generalize matrix models and generate dynamical triangulations \cite{EDTDavid,EDTAmbjorn} in any dimension 
(recently a matrix model generating dynamical triangulations is dimension three has been proposed \cite{Fukuma:2015xja}).
The continuum limit of dynamical triangulations is reached when tunning the coupling constants to some critical values \cite{Kazakov:1985ds,David:1984tx,
critical,uncoloring}. 
In this regime the number of simplices in the triangulation diverges. Sending simultaneously the volume of the individual  simplices to 
zero one obtains a continuous phase of infinitely refined random spaces \cite{DiFrancesco:1993nw,EDTAmbjorn}.

Tensor and matrix models can be regarded as field theories with no kinetic term. The fields are a 
pair of complex conjugated tensors $\mbT_{a^1\dots a^D}, \bar \mbT_{a^1\dots a^D} $ (where $a^c = 1,\dots N $) transforming under the external tensor product of 
$D$ fundamental representations of the unitary group ${\cal U}(N)$:
\[
 \mbT'_{a^1\dots a^D} = \sum_{b^1,\dots b^D=1}^N U^{(1)}_{a^1b^1} \dots U^{(D)}_{a^Db^D}  \mbT_{b^1\dots b^D} \;, \qquad
 \bar \mbT'_{a^1\dots a^D} = \sum_{b^1,\dots b^D=1}^N \bar U^{(1)}_{a^1b^1} \dots \bar U^{(D)}_{a^Db^D}  \bar \mbT_{b^1\dots b^D} \;.
\]
The action of tensor models is a function of $\mbT$ and $\bar \mbT$ which is 
invariant under unitary transformations. 
Like matrix models, tensor models are known to possess a $1/N$ expansion \cite{expansion1,expansion2,expansion3,expansioin6,expansioin5}.
The leading order in the $1/N$ expansion for matrices is given by the planar triangulations.
For tensor models, the leading order in $1/N$ is given by a class of triangulations called \emph{melonic} \cite{critical}. The melonic triangulations 
represent topological spheres in any dimensions, but one should emphasize that only a subclass of triangulations of the sphere are melonic. 
In particular the family of melonic triangulations is exponentially bounded \cite{critical}. It is not yet known whether the family of general triangulations of the sphere in dimension higher than three 
is exponentially bounded or not \cite{Rivasseau:2013bma}.

In this paper we consider a tensor model with a quartic interaction. We first rewrite this model, using the intermediate field representation \cite{Nguyen:2014mga,expansioin6}, as 
a coupled multi matrix model. Examining the vacua of this multi matrix model we find for small coupling a unique 
vacuum state invariant under unitary transformations which we call \emph{the melonic vacuum}. We derive the effective theory for the fluctuations around 
the melonic vacuum. The translation to the melonic vacuum is not a phase transition as no symmetry gets broken (the melonic vacuum possesses the full invariance of the theory). 
We show that this translation completely resums the leading order in $1/N$:
the effective theory for the perturbation fields around the melonic vacuum contributes only to lower orders in $1/N$. We show furthermore that, at each order in 
$1/N$, only a finite number of graphs of the effective theory contribute. 
Finally, we show that the criticality of the leading order contribution in $1/N$, which is the dynamical triangulation continuum limit, corresponds precisely to 
a zero eigenvalue in the effective mass matrix for the fluctuation field. 

We therefore prove that the continuum limit in dynamical triangulations is a phase transition in tensor models. This phase transition presumably corresponds to the spontaneous breaking of the unitary invariance
of the theory. In order to establish this rigorously, one needs in the future to identify the appropriate non unitary invariant vacuum state and study perturbations around it.

The situation is quite different for matrix models. A procedure similar to the one we use here for tensors can not be used for matrices: so far at least, the action of no
matrix model could be translated to an explicit ``planar vacuum state'' such that the effective theory around it contributes only a finite number of graphs at each 
order in $1/N$. 

A very interesting question is whether similar results hold for group field theories \cite{Oriti:2011jm,Baratin:2011aa,tt2}
and tensor field theories \cite{Rivasseau:2013uca,BenGeloun:2011rc,Geloun:2014kpa,Samary:2014oya,Lahoche:2015ola}.
Such theories have a genuine renormalization group flow \cite{BenGeloun:2012yk,Carrozza:2014rba} and it would be very interesting to test 
whether the phase transition we identify in this paper is associated to an infrared fixed point of this flow in which case 
tensor field theories and group field theories would naturally generate bound states of extended geometry.

This paper is organized as follows. As our results are somewhat technical, we start by discussing in section \ref{sec:discussion} their interpretation.
In section \ref{sec:model} we introduce the quartic tensor model and in section \ref{sec:translation} we discuss the effective theory around an
invariant field configuration. In section \ref{sec:multimaps} we introduce the class of Feynman graphs (combinatorial maps) relevant
for the Feynman expansion of the effective theory which we detail in section \ref{sec:Feynrules}. 
Finally in section \ref{sec:translatevacuum} we consider the translation to the invariant vacuum state and detail the $1/N$ expansion of the fluctuation field. 

\section{Discussion}\label{sec:discussion}

The partition functions of tensor and matrix models are generating functions for dynamical triangulations \cite{DiFrancesco:1993nw,EDTDavid,EDTAmbjorn}.
Denoting $\Delta$ a (strongly) connected $D$ dimensional triangulation and ${\rm Top}(\Delta)$ the number of $D$ simplices in $\Delta$,
the logarithm of the partition function (the free energy) writes as: 
\[
  W =   \frac{1}{N^D} \ln Z = \sum_{ \Delta } \;  g^{ {\rm Top}(\Delta)  } N^{ - \frac{2}{(D-1)!}\omega(\Delta)} \;,
\]
where $N$ denotes the size of the matrices (or tensors) and $\omega(\Delta)$ is a non negative integer. In the case of matrices $D=2$,
$\omega(\Delta)$ is the genus of the triangulations \cite{DiFrancesco:1993nw}, while in higher dimensions \cite{expansion1,expansion3} $\omega(\Delta)$ is the \emph{degree}.
While the degree is not a topological invariant, it serves to organize the series defining $W$ as a series in $1/N$. 

In the large $N$ limit only a subclass
of triangulations contribute: the planar triangulations in dimension two and the melonic triangulations in higher dimension. Both families are 
exponentially bounded. Similar results hold order by order in $1/N$ and it can be shown that \cite{DiFrancesco:1993nw,GurSch,Fusy:2014rba} summing the families
of triangulations at fixed order in $1/N$, the free energy becomes (up to non universal analytic pieces in $g$ which
can be disposed of by a suitable number of derivatives): 
\[
  W \sim  d_0 (g_c-g)^{\nu_0} +\sum_{q\ge 1} \frac{1}{N^q} d_q (g_c-g)^{\nu_q} \;.
\]
While the critical value $g_c$ depends on the details of the model, all the series at fixed order in $1/N$ become critical at \emph{the same critical value} $g_c$. 
The (non integer) critical exponents $\nu_0, \nu_q$ are, to a large extent, universal for a given dimension.

Considering equilateral simplices of volume $V_D$, 
the physical volume of a triangulation is proportional to the number of $D$ simplices, and the average physical volume is:
\[
 {\rm Vol}  = V_D \Braket{  {\rm Top}(\Delta) }  = V_D \; g\partial_g  \ln W  \sim_{g\to g_c} \frac{V_D}{g-g_c} \;.
\]
The continuum limit of dynamical triangulation is reached by tuning $g\to g_c$ and sending at the same time $V_D\to 0$ while keeping the average physical volume fixed.
In this regime one obtains a phase of random infinitely refined geometries. 
One can study more involved matrix and tensor models, adding for instance coupling constants for lower dimensional simplices. The various continuum phases of the model will then be spanned 
by these new couplings. 

Dynamical triangulations can also be studied \cite{EDTAmbjorn,Ambjorn:2012jv} with no reference to matrix and tensor models. In this case one restricts the family of triangulations
by some prescription. For instance one can chose to restrict to triangulations of the sphere, or to foliated triangulations \cite{Ambjorn:2012jv}. Of course one should check that the 
restricted family of triangulations is exponentially bounded (in order to be able to obtain a continuum limit) but, due to technical difficulties, this is never done in practice.
Still, one can learn numerous lessons from numerical studies. 

We show in this paper that a certain tensor model can be recast as a coupled multi matrix model for $D$ matrices $H^c$, with $c=1,\dots D$.
We then show that, in this new formulation, the vacuum state of the model is not $H^c=0$, but it is given at small coupling constant $g$ 
by the invariant configuration:
\[
 H^c = a_0 {\bf 1} \;, \qquad a_0= \frac{1-\sqrt{1-4Dg^2}}{2Dg} \;.
\]
We translate to the vacuum state and write the partition function of the model as some explicit factor multiplying the partition function of the effective theory of the 
perturbation fields $M^c$ around the vacuum state:
 \begin{align}\label{eq:summary}
   Z^{(4)} & =   e^{ - N^D \left(  \frac{D a_0^2 }{2}    +    \ln (1 - g Da_0) \right) }    (1-a_0^2)^{\frac{D-1}{2}}      \int [dM^c] \; e^{-S(M)  } \;,  \crcr
S ( M) & =  \frac{1}{2}   N^{D-1}(1-a_0^2) \sum_{c=1}^D \Tr_c \left[ M^cM^c \right] -  \frac{1}{2} N^{D-2} \frac{D-1}{D} a_0^2  \left( \sum_{c=1}^D \Tr_c[M^c] \right)^2  - Q(M)  \; .
 \end{align}
We prove on the one hand that the integral over the perturbation fields is subleading in $1/N$, namely:
\[
  - W  =    - \frac{1}{N^D} \ln Z^{(4)} = \frac{D}{2} a_0^2 + \ln (1 - g Da_0)  + O\left( \frac{1}{N}\right) \;,
\]
and on the other that the mass matrix of the effective theory for the perturbation $M^c$, $\frac{\partial^2 S (  M)}{ \partial  M \partial   M} \big{\vert}_{M=0}$ 
has eigenvalues $ N^{D-1}(1-a_0^2)$ with degeneracy $ND^2-1$ and $ N^{D-1} (1-Da_0^2)$ with degeneracy $1$.

{\it Dynamical triangulations continuum limit.} The dynamical triangulations continuum limit is obtained from the explicit factor in Eq. \eqref{eq:summary}.
The leading order free energy writes in terms of $g$ as:
\begin{align*}
 - W_{\rm LO}  = \frac{D}{2} a_0^2 + \ln (1 - g Da_0) = -\frac{1}{2} + \frac{1-\sqrt{1-4Dg^2} }{4Dg^2} - \ln \left(    \frac{1-\sqrt{1-4Dg^2} }{2Dg^2}\right) \; ,
\end{align*}
and becomes critical for the critical constant:
\[
  \boxed{ g_c =\frac{1}{2\sqrt{D} }  \; .} 
\]
In the critical regime $  \frac{g^2}{g_c^2} =  4Dg^2 = 1-\epsilon $ one obtains:
\begin{align*}
 - W_{\rm LO}  = -\frac{1}{2} + \frac{1-\sqrt{\epsilon}}{1-\epsilon} - \ln \left( 2 \frac{1- \sqrt{\epsilon}}{ 1-\epsilon} \right)
 = \frac{1}{2} +\ln 2 +\frac{1}{2} \epsilon -\frac{2}{3} \epsilon^{3/2} + O(\epsilon^2) \;,
\end{align*}
hence the critical exponent $ \nu_0= \frac{3}{2}$.

{\it Field theory phase transition.} The effective field theory for the perturbation field $M^c$ (which is subleading in $1/N$) undergoes a phase transition when the mass matrix has a zero eigenvalue,
 that is for $1-Da_0^2=0$. At this point the invariant vacuum is no longer stable and the symmetry of the theory is broken. This occurs for the coupling constant $g_{p.t.}$ solution of:
 \[
  \frac{1}{\sqrt{D}} = \frac{1-\sqrt{1-4Dg_{p.t.}^2}}{2Dg_{p.t.}} \Rightarrow \sqrt{1-4Dg_{p.t.}^2} = (1-4Dg_{p.t.}^2  ) \Rightarrow \boxed{ g_{p.t.} = \frac{1}{2\sqrt{D}} \; . }
 \]

We conclude that, in tensor models, the dynamical triangulations continuum limit is a phase transition associated to a breaking of the symmetry (unitary symmetry and color permutation) 
of the model.

We stress that the phase transition we identify in this paper is of an entirely different nature form the phase transitions studied so far in the context of 
dynamical triangulations and matrix models:
\begin{itemize}
 \item {\it Dynamical triangulations.}
In dynamical triangulations \cite{EDTAmbjorn,Ambjorn:2012jv} one takes the continuum limit 
and obtains various continuum phases. In Euclidean dynamical triangulations \cite{EDTAmbjorn} one typically obtains two continuum phases (crumpled and branched polymer), 
while in causal dynamical triangulations \cite{Ambjorn:2012jv} one obtains three continuum phases. 
One is then interested in the transition between the continuum phases. These transitions have nothing to do with the phase transition from discrete to continuum
we identify in this paper.

\item {\it Matrix models.} In matrix models \cite{Marino:2012zq} one studies the transition from a one cut 
resolvent to a multi cut resolvent. The phase transition we identify in this paper is very different, and occurs in a completely different range of parameters. 
In matrix models the continuum limit is a change of behavior of the spectral density of a one cut solution 
(one passes from a square root law to a $3/2$ power law for the distribution of eigenvalues). It is this change of behavior which we show in this paper is a genuine phase transition 
in the case of tensor models.
\end{itemize}

\section{The quartic melonic model}\label{sec:model}

A detailed introduction to the general framework of tensor models can be found in \cite{uncoloring} or
\cite{universality}, and an introduction to the quartic melonic model in \cite{Nguyen:2014mga} and \cite{expansioin6}.
An important notion for tensor models is the notion of $D$-colored graph \cite{review}. Such graphs are dual to $(D-1)$-dimensional triangulations \cite{review}. 
\begin{definition}
 A \emph{bipartite edge $D$-colored graph} ($D$-colored graph for short) $\cB$ is a graph with vertex set $\cV(\cB)$ and edge set $\cE(\cB)$ such that:
 \begin{itemize}
  \item the vertex set $\cV(\cB)$ is the disjoint union of the set $\cV^w(\cB)$ of white vertices and the set $\cV^b(\cB)$ of black vertices, with $|\cV^w(\cB)|=|\cV^b(\cB)| = k(\cB)$.
  \item the edge set $\cE(\cB)$ is the disjoint union of $D$ sets $\cE^{c}(\cB)$ of edges of color $c$, with $c\in \{1,\dots D\}$ such that 
  any edge connects a black and a white vertex.
  \item all the vertices have degree $D$ and all the edges incident to a vertex have distinct colors.
 \end{itemize}
\end{definition}

The fields in a tensor model are a pair of complex conjugated rank $D$ tensors ${\mbT}_{a^1 \dots a^D},\ \bar{\mbT}_{a^1\dots a^D}$, with $a^c= 1\dots N$.
We call the position $c$ of an index its \emph{color} and denote the set of colors $\cD=\{1 \dots D\}$. 
We denote $a^\cD = (a^c, c\in\cD)$. The tensors $\mbT$ and $\bar \mbT$ transform under the external tensor product of 
$D$ fundamental representations of the unitary group ${\cal U}(N)$:
\[
 \mbT'_{a^{\cD}} = \sum_{b^1,\dots b^D=1}^N U^{(1)}_{a^1b^1} \dots U^{(D)}_{a^Db^D}  \mbT_{b^{\cD}} \;, \qquad
 \bar \mbT'_{  a^{\cD}} = \sum_{ b^1,\dots   b^D=1}^N \bar U^{(1)}_{  a^1   b^1} \dots \bar U^{(D)}_{  a^D   b^D}  \bar \mbT_{ b^{\cD}} \;.
\]

Starting from $ \mbT $ and $\bar \mbT$ one can build a class of invariant polynomials called \emph{trace invariants}. They are in bijection with bipartite edge $D$-colored graphs.
The trace invariant associated to $\cB$ is obtained by associating to every white vertex $i$ of $\cB$ a tensor $\mbT_{a^{\cD}_i}$,
to every black vertex $\bar j$ of $\cB$ a complex conjugated tensor $\bar \mbT_{\bar a^{\cD}_{\bar j}}$ and 
to every edge $e^c=(i,\bar j)$ of color $c$ of $\cB$, a contraction of the indices of color $c$ of the two tensors associated to its end vertices $i$ and $\bar j$:
\[
 \Tr_{\cB}(\mbT,\bar \mbT) = \left( \prod_{j=1}^{k(\cB)} \mbT_{a_j^{\cD} } \right) \left( \prod_{\bar j=1}^{k(\cB)} \mbT_{ \bar a_{\bar j}^{\cD} } \right) 
  \left( \prod_{e^c = (i, \bar j)\in \cE(\cB)} \delta_{a_i^c \bar a_{\bar j}^c} \right)
\]

The melonic quartic tensor model the probability measure:
\begin{align}\label{eq:4model}
d\mu^{(4)} &=  \frac{1}{Z^{(4)}(\lambda)}\left(  \prod_{n^{\cD}}  N^{D-1} \frac{ d \bar \mbT_{n^{\cD}}  d\mbT_{n^{\cD}} }{2 \imath \pi} \right) \; 
e^{ -N^{D-1} \left(  \sum_{n^{\cD} \bar n^{\cD}} \bar \mbT_{ \bar n^{\cD} }  \delta_{ \bar n^{\cD} n^{\cD}} 
  \mbT_{ n^{\cD}}   - \frac{g^2}{2} \sum_{ c \in \cD } V^{(4)}_c ( \bar \mbT,  \mbT  )    
   \right) }  \;, \crcr
    V^{(4)}_c( \bar \mbT,  \mbT  )    &= \sum_{\bar n^{\cD} n^{\cD} m^{\cD} \bar m^{\cD}}
        \left( \bar \mbT_{\bar n^{\cD}} \delta_{\bar n^{\cD \setminus \{c\}} n^{\cD\setminus \{c\}}} \mbT_{n^{\cD}} \right)    \delta_{\bar n^{c} m^{c} } \delta_{\bar m^{c} n^{c}}
        \left( \bar \mbT_{\bar m^{\cD}} \delta_{\bar m^{\cD \setminus \{c\}} m^{ \cD \setminus \{c\}} }\mbT_{m^{\cD}} \right) \; ,
\end{align}
where, for any set $\cC \subset \cD$, we denote $ \delta_{\bar n^{\cC} n^{\cC}} = \prod_{c\in \cC} \delta_{\bar n^c n^c}$. The model is stable for $g^2<0$.
Each invariant $ V^{(4)}_c( \bar \mbT,  \mbT  )  $ can be represented as
a $D$-colored graph with two white and two black vertices and signifies a gluing of four $D$-simplices, two positively oriented and two negatively oriented. The coupling 
$g$ counts the number of positively oriented simplices. When evaluating the partition function in perturbation theory one obtains Feynman graphs having $(D+1)$ colors. Each 
such graph is dual to a $D$-dimensional triangulation.

The generating function of the moments of $\mu^{(4)}$ is:
\begin{align*}
  Z^{(4)}(J,\bar J) &= \int d\mu^{(4)} \;\;  e^{      \sum  \bar {\mathbb T}_{   \bar a^{\cD} } \bar  J_{  \bar  a^{\cD} }
  +\sum    J_{   a^{\cD}   }     {\mathbb T} _{  a^{\cD}  }   } \; ,
  \end{align*}
 The partition function is the generating function at zero external sources and is denoted by $  Z^{(4)}\equiv Z^{(4)}(0,0) $. 
The \emph{cumulants} (connected moments) of $\mu^{(4)}$ can be computed as the derivatives of the logarithm of $Z$:
\[
 \kappa_{2k} \bigl({\mathbb T}_{ a_1^{\cD} }, \bar {\mathbb T}_{ \bar a^{\cD}_{\bar 1} } ,
  \dots {\mathbb T}_{   a^{\cD}_k    } , \bar {\mathbb T}_{   \bar a^{\cD}_{\bar k} } \bigr)  = \frac{\partial^{(2k)}} {  \partial J_{    a^{\cD}_1   } \partial \bar J_{  \bar a^{\cD}_{\bar 1}  }
    \dots   \partial J_{   a^{\cD}_k   }\partial \bar J_{    \bar a^{\cD}_{\bar k}   } 
   } \ln  Z^{(4)}(J,\bar J )  \Big{\vert}_{\genfrac{}{}{0pt}{}{J =0}{\bar J =0} } \;.
\]
The cumulants are linear combinations of \emph{trace invariant operators} \cite{universality}:
\[
 \kappa_{2k } [ {\mathbb T}_{  a^{\cD}_{ 1 } } ,  \bar {\mathbb T}_{  {\bar a}^{\cD}_{ \bar 1  } } \dots  {\mathbb T}_{   a^{\cD}_{  k   } },
 \bar {\mathbb T}_{  \bar a^{\cD}_{ \bar k}   } ]
  = \sum_{ \cB, \; k(\cB)=k}  {\mathfrak K} (  \cB  ,\mu^{(4)}  )  \prod_{e^c = (i,\bar j) \in \cE(\cB)} \delta_{a^c_i\bar a_{\bar j}^c } \;,
\]
where the sum runs over all the $D$-colored graphs \cite{review} $\cB $ with $k(\cB)$ white vertices.
Somewhat abusively, we call the coefficient $   {\mathfrak K} (  \cB  ,\mu^{(4)}  ) $ also a cumulant
and we define the \emph{rescaled cumulant} corresponding to the graph $\cB$:
\[
 K(\cB, \mu^{(4)}) \equiv \frac{  {\mathfrak K} (  \cB  ,\mu^{(4)}  )   }{ N^{D- 2 k(\cB) (D-1) -C(\cB) } } \; ,
\]
where $C(\cB)$ denotes the number of connected components of $\cB$.

\subsection{The intermediate field representation}
  
The intermediate field representation of the quartic melonic models has already been discussed in detail in \cite{Nguyen:2014mga} and \cite{expansioin6}.
Let us recall here the main steps in its derivation. Each quartic interaction in expressed in the intermediate field 
(Hubbard Stratonovich) representation as an integral over an auxiliary $N \times N$ hermitian matrix:
\[
 e^{ N^{D-1} \frac{g^2}{2} \sum_{ c \in \cD } V^{(4)}_c ( \bar \mbT,  \mbT  ) }  
 =\int [dH_c]\;   e^{-\frac{1}{2} N^{D-1} \Tr[H^cH^c] +g N^{D-1}\sum_{n^c \bar n^c} H^c_{\bar n^c n^c} \sum_{n^{\cD \setminus \{ c \} } \bar n^{\cD \setminus \{ c \} }} 
 \left( \bar \mbT_{\bar n^{\cD}} \delta_{\bar n^{\cD \setminus \{c\}} n^{\cD\setminus \{c\}}} \mbT_{n^{\cD}} \right) 
 } \;, 
\]
where the integral over $H^c$ is normalized to $1$ for $g=0$. As there are $D$ quartic melonic interactions, one obtains $D$ intermediate matrix fields $H^c$. 
Each $H^c$ is an $N \times N$ matrix and the indices of $H^c$ have the color $c$.
The integral over $\mbT$ and $\bar \mbT$ is now Gaussian and can be computed explicitly to obtain:
\begin{align*}
  Z^{(4)}(J,\bar J)  & =  \int \left( \prod_{c=1}^D [dH^{c}] \right) 
    e^{ - S(H)+\frac{1}{N^{D-1}} JR(H) \bar J } \; , \crcr    
    S(H) &= \frac{1}{2} N^{D-1}  \sum_{c=1}^{ D}   \Tr_{c}[H^{c}H^{c}]  -  \Tr_{\cD} \left[ \ln R(H) \right]  \;,
\end{align*}
where $\Tr_c$ is the trace over the index of color $c$, $\Tr_{\cD}$ is the trace over all the indices, and the \emph{resolvent} is:
\[
 R(H)  =  \frac{1}{  {\bf 1}^{\cD}- g \sum_{ c=1}^D  H^{c} \otimes {\bf 1}^{\cD\setminus \{c\}}   } \; ,
\]
where ${\bf 1}^c$ is the identity matrix with indices of color $c$ and for any subset $\cC\subset \cD$ we denote ${\bf 1}^{\cC} = \bigotimes_{c\in \cC} {\bf 1}^{c}$. 

In the rest of this paper $g$ will be considered real, positive, and smaller than $\frac{1}{2\sqrt{D}}$. Observe that in this 
range of the coupling constant the resolvent can be singular. However one can still make sense of the integral defining $ Z^{(4)}(J,\bar J)$ 
by modifying the contour of integration for the field $H^c$ in the complex plane so as to avoid the singularities.

\subsection{Equations of motion}

Taking into account that:
\[  \frac{\partial}{\partial H^{c}_{\bar a^{c}  b^{c} } } \Tr_{\cD} \left[ \ln  R(H)   \right] 
= 
  g\; \sum_{ \bar p^{\cD } q^{\cD } }
   R(H)_{ q^{\cD}  \bar p^{\cD }   } \left( \delta_{\bar p^{\cD\setminus \{c\} } q^{\cD \setminus \{c\} } } \delta_{\bar p^{c} \bar a^{c} } 
    \delta_{q^{c} b^{c}} \right) \; ,
\]
the classical equations of motion in the intermediate field representation write:
\[
0=\frac{\partial S(H)}{\partial H^c_{ b^ca^c  }} =  N^{D-1} H^c_{a^cb^c} - g \; \sum_{ p^{\cD } q^{\cD } }
   R(H)_{ q^{\cD}   p^{\cD }   } \left( \delta_{  p^{\cD\setminus \{c\} } q^{\cD \setminus \{c\} } } \delta_{ p^{c} b^{c} } 
    \delta_{q^{c} a^{c}} \right)  \;.
\] 

\emph{The $H^c=0$ configuration is not a vacuum of the theory} as, for $g\neq 0$, $H^c=0$ is not a solution of the classical equations of motion.
In order to find solutions to the classical equations of motion we express them in matrix form:
\begin{equation}\label{eq:eqmotion}
  H^c  - g \frac{1}{N^{D-1}} \Tr_{\cD \setminus \{c\} } \left[   \frac{1}{  {\bf 1}^{\cD}- g \sum_{ c'=1}^D 
   H^{c'} \otimes {\bf 1}^{\cD\setminus \{c'\}}   }     \right] =0 \; ,
\end{equation}
where $\Tr_{\cD \setminus \{c\} } $ denotes the trace over the indices with colors in $\cD \setminus \{c\} $. 
A simple inspection of these equations reveals that for $g$ small enough the stable vacuum of the theory is 
invariant under conjugation by the unitary group and invariant under color permutation, $H^c = a {\bf 1}$. 
The invariant solutions of eq. \eqref{eq:eqmotion} 
are then obtained for $a$ satisfying the self consistency equation:
\[
 a  = \frac{ g}{1 - g Da } \; .
\]

This self consistency equation has two solutions: 
\begin{itemize}
 \item {\it The melonic vacuum.} A first solution, denoted $a_0$ and called the \emph{melonic vacuum}, is the sum of a power series in $g$:
 \[
  a_0 =  \frac{1 - \sqrt{1-4Dg^2 }}{2g D } = g \sum_{n\ge 0} \frac{1}{(n+1)!} \binom{2n}{n} (Dg^2)^n \; .
\]
We will see below that this solution is the stable vacuum of the theory for $g$ small enough.
In the region of interest ($g\in \mathbb{R_+},\; 4Dg^2 <1 $), $a_0$ is real, positive, and bounded by:
\[
 a_0  =   \frac{1 - \sqrt{1-4Dg^2}}{ 2D g } = \frac{2 g}{ 1 +  \sqrt{1-4Dg^2} }
  \le \frac{1}{\sqrt{D}} \;,
\]
as $ \frac{2 g }{ 1 +  \sqrt{1-4Dg^2} }$  is increasing with $g$ and attains its maximum for $g = \frac{1}{2\sqrt{D}} $.

\item {\it The instanton.} The second solution, denoted $a_{\rm inst}$, is an instanton solution:
\[
  a_{\rm inst} =  \frac{1 + \sqrt{1-4Dg^2}}{2g D } \; .
\]
As $a_0  a_{\rm inst} = \frac{1}{D}$, it follows that in the region of interest ($g\in \mathbb{R_+},\; 4Dg^2 <1 $), $a_{\rm inst}$
is real, positive, and $a_{\rm inst} \ge \frac{1}{\sqrt{D}} $.
\end{itemize}

For $g = \frac{1}{2\sqrt{D}}$ the two solutions collapse: $a_0|_{g = \frac{1}{2\sqrt{D}} }=a_{\rm inst}|_{g = \frac{1}{2\sqrt{D}} }=\frac{1}{\sqrt{D}}$. 

\section{Translating the intermediate field}\label{sec:translation}

In order to study the effective theory around a solution of the classical equations of motion we must translate the field to $H^c = a \mathbf{1} + M^c$.
This translation is a translation (in the complex plane) by $a$ of the contour of integration of all the diagonal entries of $H^c$. The effective theory is obtained by Taylor expanding in $M^c$. 
We first discuss a translation of $H^c$ by some arbitrary constant $a$.
We denote $M =\sum_{ c=1}^D  M^{c} \otimes {\bf 1}^{\cD\setminus \{c\}}    $ and we define:
\[
 b(a) \equiv \frac{g}{ 1 - g Da  } \; .
\]
Observe that if $g$ is real and $a$ is real, then $b(a)$ is also real. 
The melonic vacuum and the instanton are the two solutions of the self consistency equation $a = b(a)$.
We will use the shorthand notation $b$ instead of $b(a)$ but the reader should keep in mind that $b$ is a function of 
$a$. 

In terms of the perturbation field $M$, the resolvent and the trace of its logarithm write as:
\begin{align*}
 & R(a \mathbf{1}^{\cD} + M)   = \frac{1}{(1 - g Da )\mathbf{1}^{\cD}  - gM } = 
   \left( \frac{1}{ 1 - g Da  } \right)  \frac{1}{ \mathbf{1}^{\cD}  - b  M   } \; , \crcr
 & \Tr_{\cD} \left[ \ln R(a \mathbf{1}^{\cD} + M) \right]  = -N^D \ln (1 - g Da) 
 + \Tr_{\cD} \left[ \ln  \frac{1}{ \mathbf{1}^{\cD}  - b  M   }  \right] \; .
\end{align*}

After translation by $a$ the action becomes:
\begin{align*}
  \frac{1}{2} N^D D a^2 + N^{D-1} a \sum_{c=1}^D \Tr_c[M^c] + \frac{1}{2}N^{D-1}\sum_{c=1}^D \Tr_{c}[M^c M^c]  + N^D \ln (1 - g Da) - \Tr_{\cD} \left[ \ln  \frac{1}{ \mathbf{1}^{\cD}  - b  M   }  \right] \; ,
\end{align*}
and the generating function writes in terms of the perturbation field $M$:
\begin{align*}
   Z^{(4)}(J,\bar J) & =   e^{ - N^D \left(  \frac{D a^2 }{2}    +    \ln (1 - g Da) \right) } \\
   & \qquad \times   \int \left( \prod_{c=1}^D [dM^{c}] \right)  
    e^{  - \frac{N^{D-1}}{2}\sum_{c=1}^D \Tr_{c}[M^c M^c]  - a   \Tr_{\cD }[M]   +\Tr_{\cD} \left[ \ln  \frac{1}{ \mathbf{1}^{\cD}  - b  M   }  \right] + \frac{1}{N^{D-1} } 
     \frac{1}{(1 - g Da )} J \left(  \frac{1}{ \mathbf{1}^{\cD}  - b  M   } \right)\bar J
     } \; .
\end{align*}

The logarithm of the resolvent contains quadratic terms in $M$, which must be reabsorbed in the quadratic part of the action.
Let us define the subtracted vertex function:
\[
 Q(M) = \Tr_{\cD} \left[ \ln  \frac{1}{ \mathbf{1}^{\cD}  - b  M   }  \right] - \Tr_{\cD} \left[ a  M \right] - \frac{1}{2}\Tr_{\cD} \left[ \big(b  M \big)^2\right] \;,
\]
which, by definition, has no quadratic terms in $M$. The generating function is then:

\begin{align}\label{eq:Ztranslatat}
   Z^{(4)}(J,\bar J) & =   e^{ - N^D \left(  \frac{D a^2 }{2}    +    \ln (1 - g Da) \right) } \nonumber\\
   & \qquad \times   \int \left( \prod_{c=1}^D [dM^{c}] \right)  
    e^{  - \frac{N^{D-1}}{2}\sum_{c=1}^D \Tr_{c}[M^c M^c] + \frac{1}{2}\Tr_{\cD} \left[ \big( b M \big)^2\right] +   Q(M) + \frac{1}{N^{D-1} } 
     \frac{1}{(1 - g Da )} J \left(  \frac{1}{ \mathbf{1}^{\cD}  - b  M   } \right)\bar J
     } \; .
\end{align}

\paragraph*{Integral form of the subtracted vertex function.}
One can give a closed integral formula for the subtracted vertex function $Q(M)$. A Taylor expansion with integral rest yields:
\begin{align*}
  & \Tr_{\cD} \left[ \ln  \frac{1}{ \mathbf{1}^{\cD}  - b  M   }  \right]   = 
   \crcr
 & \qquad =\left\{ \partial_t  \Tr_{\cD} \left[ \ln  \frac{1}{ \mathbf{1}^{\cD}  - t b  M   }  \right] \right\}_{t=0}   
 +  \frac{1}{2} \left\{ \partial^2_t  \Tr_{\cD} \left[ \ln  \frac{1}{ \mathbf{1}^{\cD}  - t b  M   }  \right] \right\}_{t=0}  
 + \frac{1}{2} \int_0^1 dt \; (1-t)^2  \partial^3_t  \Tr_{\cD} \left[ \ln  \frac{1}{ \mathbf{1}^{\cD}  - t b  M   }  \right]\; ,
\end{align*}
and substituting the derivatives:
\[
 \partial_t^r \Tr_{\cD} \left[ \ln  \frac{1}{ \mathbf{1}^{\cD}  - t b  M   }  \right] = (r-1)! 
 \Tr_{\cD} \left[    \frac{1}{ (\mathbf{1}^{\cD}  - t b  M)^r   } \big( b M \big)^{r} \right] \; ,
\]
we obtain:
\[
   Q(M)  =  (b-a) \Tr_{\cD} \left[ M \right] + \int_0^1 dt \; (1-t)^2   \Tr_{\cD} \left[    \frac{1}{ ( \mathbf{1}^{\cD}  - t b  M  )^3 } \big( b  M \big)^{3} \right] \; ,
\]
that is $Q(M)$ is the sum of a linear piece and a piece starting at order $M^3$.

\paragraph*{Linear term.} At the price of modifying the covariance, the translation of $H^c$ by a generic $a$  introduces some univalent vertices in the theory.
Choosing $a=a_0$ comes to choosing the weight of the univalent vertex as the counterterm of the melonic graphs \cite{critical,uncoloring}, and by translating to the melonic vacuum, one subtracts the contribution 
of the \emph{entire melonic family}: as $a_0 = b(a_0)$, the linear term $\big(  b  -a  \big)    \Tr_{\cD}[M] $
which represents the univalent vertices vanishes. As expected, in this case $M^c=0$ is a solution of the classical equations of motion due to the fact that the remaining piece of 
$Q(M)$ starts at order $M^3$.

No such interpretation exists for the translation to the instanton $a_{\rm inst}$ although the linear term cancels also in that case.

\paragraph*{Free energy.} We can compute the large $N$ free energy of the model in the perturbative regime using Schwinger Dyson equations.
According to \cite{uncoloring}, the large $N$ covariance $K_2$ and the large $N$ free energy are:
\begin{align*}
 & 1 - K_2 + D g^2 K_2^2 =0\ \Rightarrow\ K_2 = \frac{  1 - \sqrt{ 1 - 4D g^2 } }{ 2D g^2 }=  \frac{ a_0}{g } \;, \crcr
 &  W_{\infty} =\lim_{N\to \infty} \left(\frac{1}{N^D} \ln Z^{(4)}(0,0) \right)= 1 + \ln ( K_2 ) -   K_2 + \frac{Dg^2}{2} K_2^2 = - \frac{Dg^2}{2} K_2^2 +  \ln ( K_2 ) \;.
\end{align*}
The large $N$ free energy writes in term of $a_0$ as:
\[
   W_{\infty} = -\frac{D}{2} a_0^2 - \ln (1 - g Da_0)  \;,
 \]
that is precisely the explicit prefactor in Eq.~\eqref{eq:Ztranslatat} for $a =  a_0$. We will re-derive this result below.

It ensues that the translation to the melonic vacuum $H^c = a_0\mathbf{1}$ resums the contribution of the melonic sector,  
and the remaining integral over $M$ (which, as we will see below, only contributes at lower orders in $1/N$) is precisely the effective theory of a perturbation field around the melonic phase. 

We note that the case of matrices, $D=2$, is quite different. A careful inspection 
shows that for matrices, even after translating to $a=a_0$, the integral over the perturbation $M$ contributes at leading order in $1/N$. This is natural because for matrices the melonic 
graphs are only a subset of the planar graphs, and one cannot subtract the contribution of the planar sector by a single counterterm.
 
\subsection{The effective theory for the fluctuation field}

The space of Hermitian $N \times N$ matrices is a $N^2$ dimensional real vector space ${\cal H}$ with 
inner product $\Braket{A|B} \equiv \Tr[ A B]$. The quadratic part of the action is a quadratic form 
on the direct sum of $D$ such Hilbert spaces $\bigoplus_{c=1}^D \cH^c$. 
Denoting a generic element in this direct sum space by:
\[
 \hat M = \begin{pmatrix}
            M^1 \\ \vdots \\ M^D 
           \end{pmatrix} \; ,
\]
the inner product in $\bigoplus_{c=1}^D \cH^c $ is:
\[
 \Braket{ \hat P| \hat  M } = \sum_{c=1}^D \Tr_c[  P^c M^c] \; .
\]

The effective theory for the fluctuation field $\hat M$ is somewhat involved due to the fact that, while the quadratic part of the action is a quadratic form on the direct sum space $\bigoplus_{c=1}^D \cH^c $,
the interaction $Q(M)$ is a trace in the \emph{tensor product} space $\bigotimes_{c=1}^D \cH^c $. This is the root of all the subtleties we will need to deal with below.
 
We now work out the effective covariance for the fluctuation field. The  quadratic part of the action of the fluctuation field in Eq.~\eqref{eq:Ztranslatat} is:
 \begin{align*}
  & N^{D-1} \sum_{c=1}^D \Tr_c[M^c M^c] - \Tr_{\cD}[  ( b M)^2] = \crcr
  & \qquad  = N^{D-1} (1-b^2)  \sum_{c=1}^D \Tr_c[ (M^c)^2] + N^{D-2} b^2 \sum_{c=1}^D \Big( \Tr_c[M^c] \Big)^2 - N^{D-2} b^2 \left( \sum_{c=1}^D  \Tr_c [M^c] \right)^2 \;.
 \end{align*}

We denote by $ \hat I_c$ the element of $\bigoplus_{c=1}^D \cH^c$ consisting in the identity matrix in ${\cal H}^c$ and zero on the other components
and by $ \hat I$ the element consisting in the identity matrix in all the ${\cal H}^c$s:
\[
 \hat  I_c = \begin{pmatrix}
             0  \\ \vdots \\{\bf 1} \\ \vdots \\ 0 
           \end{pmatrix}  \; ,
 \quad  \hat I  = \begin{pmatrix}
             {\bf 1}  \\ \vdots \\{\bf 1}  
           \end{pmatrix}  = \sum_{c=1}^D \hat I_c \; .        
\]
The orthogonal projectors on the one dimensional subspaces generated by $ \hat {I}_c$ and $ \hat {I}$ are:
\begin{align*}
 {\bf P}_{ c} \;  {\hat M} & = \frac{\Braket{{\hat I}_c| {\hat M} } }{\Braket{{\hat I}_c| {\hat I}_c }} {\hat I}_c = \frac{\Tr_c[M^c]}{N} {\hat I}_c  \Rightarrow 
 \Braket{{\hat M} | {\bf P}_{c} \;  {\hat M} } =  \frac{\Tr_c[M^c]}{N} \Tr_c[M^c]
 \crcr
 {\bf P}  {\hat M} & = \frac{\Braket{{\hat I}| {\hat M} } }{\Braket{{\hat I}| {\hat I}}} {\hat I} 
 = \frac{ \sum_{c=1}^D\Tr_c[M^c]}{DN} {\hat I} \Rightarrow    \Braket{{\hat M} | {\bf P}  {\hat M} } = \frac{ \sum_{c=1}^D\Tr_c[M^c]}{DN}
   \left( \sum_{c'=1}^D\Tr_{c'}[M^{c'}] \right) \; .
\end{align*}
Let us denote ${\bf I}$ the identity operator on $ \bigoplus_{c=1}^D \cH^c$. The quadratic part writes then as 
$ \Braket{ \hat M | {\bf O} \; \hat M} $ where the operator ${\bf O}$ is:
\[
 {\bf O} = N^{D-1}(1-b^2) \; {\bf I} + N^{D-1}b^2 \sum_{c=1}^D  {\bf P}_c - N^{D-1} Db^2 \;  {\bf P} \;. 
\]
 
The covariance of the effective theory is the operator ${\bf O}^{-1}$. In order to compute it, it is most convenient to diagonalize ${\bf O}$.
 
\begin{lemma}
The operator ${\bf O}$   is diagonalized as:
\[
 {\bf O} = N^{D-1}(1-b^2) \left[ {\bf I} -  \sum_{c=1}^D  {\bf P}_c  \right] +
 N^{D-1} \left[ \sum_{c=1}^D  {\bf P}_c - {\bf P}  \right] +  N^{D-1} (1- Db^2 )  {\bf P} \; ,
\]
where  $   {\bf I} -  \sum_{c=1}^D  {\bf P}_c   $, $ \sum_{c=1}^D  {\bf P}_c - {\bf P} $ and ${\bf P} $ 
are the projectors on the eigenspaces of ${\bf O}$. The eigenvalues of ${\bf O}$ and their degeneracies are:
\begin{align*}
\Lambda_1 =N^{D-1} (1-b^2) \;&, \quad  {\rm dim} \left( {\rm Im}\left[ {\bf I} -  \sum_{c=1}^D  {\bf P}_c  \right] \right)  = DN^2-D \; ,\crcr 
\Lambda_2 =N^{D-1} \;&,  \quad {\rm dim} \left( {\rm Im }\left[ \sum_{c=1}^D  {\bf P}_c - {\bf P}  \right] \right)  = D-1 \; ,\crcr
\Lambda_3 =N^{D-1}(1-Db^2) \;&, \quad  {\rm dim} \left( {\rm Im} \left[ {\bf P } \right]  \right)  = 1 \; .
 \end{align*}

\end{lemma}
\begin{proof} It is enough to check that the operators $   {\bf I} -  \sum_{c=1}^D  {\bf P}_c   $, $ \sum_{c=1}^D  {\bf P}_c - {\bf P} $ and ${\bf P} $ 
are three (mutually orthogonal) orthogonal projections. The three operators are obviously Hermitian and:
\begin{align*}
& {\bf P}_{c'} {\bf P}_{c} \hat M = \frac{\Tr_c[M^c]}{N }  \frac{ \Braket{\hat I_{c'} | \hat I_{c}} }{N} \hat I_{c'} = \delta_{cc'} {\bf P}_c \hat M \; , \crcr
& \left( \sum_{c=1}^D {\bf P}_c \right) {\bf P}  \hat M 
= \frac{ \sum_{c'=1}^D\Tr_{c'}[M^{c'}] }{ DN }  \sum_{c=1}^D \frac{ \Braket{\hat I_c | I} }{N} \hat I_c = {\bf P}  \hat M \; ,\crcr
&  {\bf P} \left( \sum_{c=1}^D {\bf P}_c \right) \hat M  = \sum_{c=1}^D \frac{\Tr_c[M^c]}{N} 
\frac{\Braket{{\hat I}| {\hat I}_c } }{DN} {\hat I} ={\bf P}  \hat M  \; .
\end{align*}
The rest of the proof is a straightforward computation:
\begin{align*}
  \left[ {\bf I} -  \sum_{c=1}^D  {\bf P}_c  \right]  \left[ {\bf I} -  \sum_{c=1}^D  {\bf P}_c  \right] & =
   {\bf I} -  2  \sum_{c=1}^D  {\bf P}_c + \sum_{c,c'=1}^D {\bf P}_c {\bf P}_{c'}  =  {\bf I} -  \sum_{c=1}^D  {\bf P}_c  \;,  \crcr
 \left[ \sum_{c=1}^D  {\bf P}_c - {\bf P}  \right] \left[ \sum_{c=1}^D  {\bf P}_c - {\bf P}  \right] & =  
\sum_{c,c'=1}^D {\bf P}_c {\bf P}_{c'}  -  {\bf P} \left( \sum_{c=1}^D  {\bf P}_c \right)-  \left( \sum_{c=1}^D  {\bf P}_c \right) {\bf P}+ {\bf P} \crcr
& = \sum_{c=1}^D  {\bf P}_c - {\bf P} \;, \crcr
{\bf P}^2 & = {\bf P} \;, \crcr
 \left[ {\bf I} -  \sum_{c=1}^D  {\bf P}_c  \right] \left[ \sum_{c=1}^D  {\bf P}_c - {\bf P}  \right]  & =
 \sum_{c=1}^D  {\bf P}_c -  \sum_{c,c'=1}^D {\bf P}_c {\bf P}_{c'}  - {\bf P} + \left( \sum_{c=1}^D  {\bf P}_c \right) {\bf P} =0 \; ,\crcr
  \left[ {\bf I} -  \sum_{c=1}^D  {\bf P}_c  \right] {\bf P} &= {\bf P} -  \left( \sum_{c=1}^D  {\bf P}_c \right) {\bf P} =0  \; , \crcr
 \left[ \sum_{c=1}^D  {\bf P}_c - {\bf P}  \right] {\bf P} &= \left( \sum_{c=1}^D  {\bf P}_c \right) {\bf P}  - {\bf P}^2 =0 \; .
\end{align*}

Concerning the dimensions of the images of each of these operators, by definition ${\bf P}$ is one dimensional, 
${\bf P}_c$ and ${\bf P}_{c'}$ are orthogonal if $c\neq c'$ hence $ \left[ {\bf I} -  \sum_{c=1}^D  {\bf P}_c  \right] $
is $DN^2 -D$ dimensional and finally the sum of the three projectors is the identity.

\end{proof}
 
It is now easy to compute the covariance of the effective theory by inverting ${\bf O}$:
\begin{align*}
 {\bf O}^{-1} & = \frac{1}{N^{D-1}(1-b^2)} \left[ {\bf I} -  \sum_{c=1}^D  {\bf P}_c  \right] +
 \frac{1}{N^{D-1}} \left[ \sum_{c=1}^D  {\bf P}_c - {\bf P}  \right] +  \frac{1}{N^{D-1} (1- Db^2 )}  {\bf P} \crcr
 & = \frac{1}{N^{D-1}(1-b^2)}  {\bf I} 
  -\frac{b^2}{N^{D-1}(1-b^2)}  \left( \sum_{c=1}^D  {\bf P}_c \right)  +  \frac{Db^2}{N^{D-1}(1-Db^2)}  {\bf P} \;.
\end{align*}

Observing that the normalized Gaussian integral with covariance $C$ for any field $\phi$ can be written as a differential operator \cite{salmhofer1999renormalization}:
\[
 \int \frac{ [d\phi]}{ \sqrt{\det C}} \; e^{ - \frac{1}{2} \phi C^{-1} \phi +F(\phi)} \equiv \left[e^{\frac{1}{2} \frac{\partial}{\partial \phi} C  \frac{\partial}{\partial \phi}}  e^{F(\phi) }\right]_{\phi=0} \;, 
\]
and, taking into account the normalization of the Gaussian integral, 
the generating function in Eq.~\eqref{eq:Ztranslatat} becomes:
\begin{align}\label{eq:Ztranslatatdiff}
  Z^{(4)}(J,\bar J)   =  \frac{e^{ - N^D \left(  \frac{D a^2 }{2}    +    \ln (1 - g Da) \right) } }{\sqrt{  (1-b^2)^{DN^2-D}   (1-Db^2)   }  }
 \Big[ e^{\frac{1}{2} \Braket { \frac{\partial}{\partial \hat M} | {\bf O}^{-1} \frac{\partial}{\partial \hat M}  } } 
   \;  e^{    Q(M) + \frac{1}{N^{D-1} } 
     \frac{1}{(1 - g Da )} J \left(  \frac{1}{ \mathbf{1}^{\cD}  - b M   } \right)\bar J} \Big]_{M^c=0}   \; .
 \end{align}
Using the explicit form of the projectors ${\bf P}_c$ and ${\bf P}$, the differential operator in the above equation writes in terms of derivatives with respect to $M^c$ as:
\begin{align*}
   \Braket { \frac{\partial}{\partial \hat M} | {\bf O}^{-1} \frac{\partial}{\partial \hat M}  } & =  \frac{1}{N^{D-1}(1-b^2)}  \sum_{c=1}^D \Tr_c\left[ \frac{\partial}{\partial M^c} \frac{\partial}{\partial M^c} \right] 
  -\frac{b^2}{N^{D}(1-b^2)}  \sum_{c=1}^D \Tr_c\left[  \frac{\partial}{\partial M^c} \right]   \Tr_c\left[  \frac{\partial}{\partial M^c} \right] + \crcr
 & \qquad \qquad  +  \frac{b^2}{N^{D}(1-Db^2)}    \sum_{c,c'=1}^D \Tr_c\left[  \frac{\partial}{\partial M^c} \right]   \Tr_{c'}\left[  \frac{\partial}{\partial M^{c'}} \right]  \; .
\end{align*}

The effective theory can be significantly simplified. The crucial observation is that the differential operator $ \Braket { \frac{\partial}{\partial \hat M} | {\bf O}^{-1} \frac{\partial}{\partial \hat M}  } $
acts on functions which depend only on $M = \sum_{c=1}^D M^c \otimes {\bf 1}^{\cD \setminus \{c\} }$. We observe that the action of an operator:
 \begin{align*}
&  \Tr_c\left[ \frac{\partial}{\partial M^c} \right] \left( \sum_{c=1}^D M^c \otimes {\bf 1}^{\cD\setminus\{c\} } \right)^n 
=   \sum_{n_1,\dots n_D \ge 0}^{n_1+\dots n_D=n} \frac{n!}{n_1!\dots n_D!} n_c  \; (M^1)^{n_1} \otimes \dots (M^c)^{n_c-1} \dots \otimes (M^D)^{n_D} \crcr
& \qquad =   n \sum_{n_1,\dots n_D \ge 0}^{n_1+\dots n_D=n-1} \frac{(n-1)!}{n_1!\dots n_D!}   (M^1)^{n_1} \otimes \dots (M^c)^{n_c} \dots \otimes (M^D)^{n_D}  
= n \left( \sum_c M^c \otimes {\bf 1}^{\cD\setminus\{c\} } \right)^{n-1} \;,
\end{align*}
\emph{does not depend} on the color $c$. It follows that 
all the terms $ \Tr_c[ \frac{\partial}{\partial M^{c}}]$ can be replaced by the symmetrized operator $ \frac{1}{D} \sum_{c=1}^D \Tr_c[ \frac{\partial}{\partial M^{c}}] $
hence the operator $\Braket { \frac{\partial}{\partial \hat M} | {\bf O}^{-1} \frac{\partial}{\partial \hat M}  }$ can in fact be replaced with:
 \begin{align*}
 & \Braket { \frac{\partial}{\partial \hat M} | \tilde {\bf O}^{-1} \frac{\partial}{\partial \hat M}  } = \crcr
 & \quad = \frac{1}{N^{D-1}(1-b^2)}  \sum_{c=1}^D \Tr_c\left[ \frac{\partial}{\partial M^c} \frac{\partial}{\partial M^c} \right]   
 + \frac{b^2 D (D-1)} {N^{D}(1-b^2)(1-Db^2)}  \left( \frac{1}{D} \sum_{c=1}^D\Tr_c\left[  \frac{\partial}{\partial M^c} \right]   \right)
 \left( \frac{1}{D} \sum_{c=1}^D\Tr_c\left[  \frac{\partial}{\partial M^c} \right]   \right) \;.
\end{align*}

The new operator, $\tilde {\bf O}^{-1}$, which defines the covariance of the effective theory can be written in terms of ${\bf I}$ and ${\bf P}$ only: 
\[
\tilde {\bf O}^{-1} =   \frac{1}{N^{D-1}(1-b^2)} {\bf I} +   \frac{b^2 (D-1)} {N^{D-1}(1-b^2)(1-Db^2)} {\bf P}   =
  \frac{1}{N^{D-1}(1-b^2)} ( {\bf I} - {\bf P} )  +   \frac{1} {N^{D-1} (1-Db^2)} {\bf P}    \;. \]

Correspondingly, the quadratic part of the action becomes:
\begin{align*}
& \Braket{\hat M | \tilde {\bf O}  \hat M} = \Braket{ \hat M | \Big( N^{D-1}(1-b^2)  {\bf I} - N^{D-1} (D-1)b^2  {\bf P}  \Big)\hat M }  \crcr
& \qquad  = N^{D-1}(1-b^2) \sum_{c=1}^D \Tr_c \left[ M^cM^c \right] - N^{D-2} \frac{D-1}{D} b^2  \left( \sum_{c=1}^D \Tr_c[M^c] \right)^2 \;. 
\end{align*}

With this remark, the generating function can then be rewritten in the simpler form:
\begin{align}\label{eq:smecher}
   Z^{(4)}(J,\bar J) = & \frac{e^{ - N^D \left(  \frac{D a^2 }{2}    +    \ln (1 - g Da) \right) } }{\sqrt{  (1-b^2)^{DN^2-D}   (1-Db^2)  }  } \crcr
 & \quad 
 \Big[ e^{\frac{1}{2} \frac{1}{N^{D-1}(1-b^2)}  \sum_{c=1}^D \Tr_c\left[ \frac{\partial}{\partial M^c} \frac{\partial}{\partial M^c} \right]  
 +\frac{1}{2}\frac{b^2 D (D-1)} {N^{D}(1-b^2)(1-Db^2)}  \left( \frac{1}{D} \sum_{c=1}^D\Tr_c\left[  \frac{\partial}{\partial M^c} \right]   \right)
 \left( \frac{1}{D} \sum_{c=1}^D\Tr_c\left[  \frac{\partial}{\partial M^c} \right]   \right)
 } 
 \crcr
 & \qquad \qquad \times  e^{    Q(M) + \frac{1}{N^{D-1} } 
     \frac{1}{(1 - g Da )} J \left(  \frac{1}{ \mathbf{1}^{\cD}  - b M   } \right)\bar J} \Big]_{M^c=0}
     \; ,
 \end{align}
 or, in integral form:
 \begin{align}\label{eq:smecherie}
   Z^{(4)}(J,\bar J) = & e^{ - N^D \left(  \frac{D a^2 }{2}    +    \ln (1 - g Da) \right) }    (1-b^2)^{\frac{D-1}{2}}       \crcr
& \quad  \int [dM^c] e^{-\frac{1}{2}   N^{D-1}(1-b^2) \sum_{c=1}^D \Tr_c \left[ M^cM^c \right] + \frac{1}{2} N^{D-2} \frac{D-1}{D} b^2  \left( \sum_{c=1}^D \Tr_c[M^c] \right)^2   } \crcr
& \qquad \qquad  \times  e^{    Q(M) + \frac{1}{N^{D-1} } 
     \frac{1}{(1 - g Da )} J \left(  \frac{1}{ \mathbf{1}^{\cD}  - b M   } \right)\bar J}
     \; ,
 \end{align}
 where the integral over $M^c$ is normalized to $1$ for $g=0,b=0,J=\bar J=0$. 
 
 We conclude that, after a translation by $a$ of the intermediate field, the quartic tensor model becomes (up to an overall factor) 
 a field theory for the fluctuation matrix fields $M^c$ with effective action:
 \[
\boxed{   S ( M) =  \frac{1}{2}   N^{D-1}(1-b^2) \sum_{c=1}^D \Tr_c \left[ M^cM^c \right] -  \frac{1}{2} N^{D-2} \frac{D-1}{D} b^2  \left( \sum_{c=1}^D \Tr_c[M^c] \right)^2  - Q(M) \;. }
 \]

\subsubsection{Stability}

The stability of the effective theory is encoded in the effective mass matrix of the fluctuation field, which is:
\[
 \frac{\partial^2 S (  M)}{ \partial  M \partial   M} \Big{|}_{M=0} =  N^{D-1}(1-b^2) ( {\bf I} - {\bf P} )  +   N^{D-1} (1-Db^2)  {\bf P}   \;.
\]
The effective mass matrix has then two eigenvalues: $ N^{D-1}(1-b^2)$ with degeneracy $ND^2-1$ and $ N^{D-1} (1-Db^2)$ with degeneracy $1$.
Let us analyze this effective mass matrix for the melonic vacuum and the instanton.

\paragraph*{The melonic vacuum.} The melonic vacuum $a_0$ is a solution of the equation $b(a_0)=a_0$. Furthermore,
$a_0$ is bounded from above by $\frac{1}{\sqrt{D}}$ for $g\le \frac{1}{2\sqrt{D}}$
and attains this value exactly at $ g= \frac{1}{2\sqrt{D}} $. It follows that all the eigenvalues of the effective mass matrix
around the melonic vacuum are positive for $g<\frac{1}{2\sqrt{D}}$, and at precisely $g=\frac{1}{2\sqrt{D}}$ one of these eigenvalues 
becomes $0$. The melonic vacuum is therefore the stable vacuum of the theory for $g<\frac{1}{2\sqrt{D}}$ and at $g= \frac{1}{2\sqrt{D}}$ 
the effective mass matrix around this vacuum develops a zero eigenvalue, hence the effective theory for the fluctuation field undergoes a phase transition.
This phase transition corresponds to a breaking of the full symmetry of the  model (conjugation by $D$ unitary transformations and color permutation).

\paragraph{The instanton.} The instanton is also a solution of the equation $b(a_0)=a_0$, but this time $a_0\ge \frac{1}{\sqrt{D}}$.
The effective mass matrix around the instanton has at least a negative eigenvalue in the region $g<\frac{1}{2\sqrt{D}}$, hence in this region the instanton is 
unstable. 
 
\section{Edge multicolored maps}\label{sec:multimaps}

In this section we introduce the notion of edge multicolored maps. We will show in the next section that the Feynman expansion
for the fluctuation field $M$ is indexed by these objects.

\begin{definition}
A \emph{combinatorial map with external cilia} is:
\begin{itemize}
\item a finite set of half edges $\cS$ which is the disjoint union of the set of \emph{internal half edges} $\cS_{\rm  int}$  and the set of external \emph{cilia} $\cS_{\rm ext}$.
We denote the elements of $\cS$ by $h\in \cS$.
\item a permutation $\sigma$ on $\cS$.
\item an involution $\alpha$ on the set of internal half edges $\cS_{\rm int}$ \emph{having no fixed points}.
The involution $\alpha$ on $\cS_{\rm int}$ is extended to an involution on $\cS$ (which we denote also by $\alpha$) 
\emph{with fixed points} by imposing that 
$\alpha(h)=h,\; \forall h\in \cS_{\rm ext}$, that is by imposing that all the external cilia are fixed points of $\alpha$.
\end{itemize}
\end{definition}

The cycles of the permutation $\sigma$ are the vertices of the map: the half edge $\sigma(h)$ is the successor of the half edge $h$ when turning around 
a vertex of the map. These maps have a well defined notion of \emph{faces}.

\begin{definition}
  The cycles of the permutation $\sigma\alpha $ fall in two categories:
 \begin{itemize}
  \item the cycles of $ \sigma\alpha $ which \emph{do not contain any cilium}, that is such that: \[ \alpha (h) \notin \cS_{\rm ext} \; ,\] 
   for any $h$ in the cycle are called the \emph{internal faces} of the map.
  \item the cycles of $\sigma \alpha$ which contain cilia. They are further subdivided into sequences of half edges separating two consecutive cilia:
\[  h, \sigma \alpha (h) ,\dots , (\sigma \alpha)^r (h)  \quad  {\rm such \; that} \;
\begin{cases}
 \alpha(h)\in \cS_{\rm ext} \crcr
  \forall 0<s<r , \;\; \alpha (\sigma \alpha)^s (h)   \notin \cS_{\rm ext} \crcr 
  \alpha  (\sigma \alpha)^r (h)  \in \cS_{\rm ext}
\end{cases}
\;, \]
 called \emph{external strands}.
 \end{itemize}
\end{definition}

\begin{definition}
The \emph{corners} of a combinatorial map are the couples:
\[\big( h,\sigma(h) \big), \; \forall h\in \cS \; .\]
\end{definition}

The corners can be pictured as the pieces of vertices comprised between two consecutive half edges or between a half edge and a cilium.

An edge multicolored map with cilia is a map with cilia whose edges are furthermore colored by a subset $\cC$ of colors, $\cC \subset \{1,\dots D\}$ (in the next section we will 
encounter only maps such that $\cC$ consists in a unique color, but for now we keep the discussion general).

\begin{definition}
An \emph{edge multicolored combinatorial map with external cilia} $\cM$ is:
\begin{itemize}
\item a finite set $\cS$ which is the disjoint union of the sets $\cS^{(\cC)}_{\rm  int}$ of internal half edges of the colors $\cC$  and the set $\cS_{\rm ext}$ of external \emph{cilia}:
\[
 \cS = \left(\bigsqcup_{ \cC \subset \cD } \cS^{(\cC)}_{\rm  int} \right) \bigsqcup \cS_{\rm ext}
\]
\item a permutation $\sigma$ on $\cS$.
\item for every $\cC\subset \cD$, an involution $\alpha^{(\cC)}$ on the set of internal half edges of colors $\cC$, $\cS^{(\cC)}_{\rm int}$, having no fixed points.
We extend the involutions $\alpha^{(\cC)}$ to the whole of $\cS$ by setting 
$\alpha^{(\cC)}(h)=h,\; \forall h\in \cS\setminus \cS^{(\cC)}_{\rm int}$.
\end{itemize}
\end{definition}

An example of an edge multicolored map with external cilia is presented in Figure~\ref{fig:multicoloredmap}.
\begin{figure}[ht]
\begin{center}
\psfrag{c}{$c$}
\psfrag{c1}{$c_1$}
\psfrag{cc1}{$cc_1$}
 \includegraphics[width=5cm]{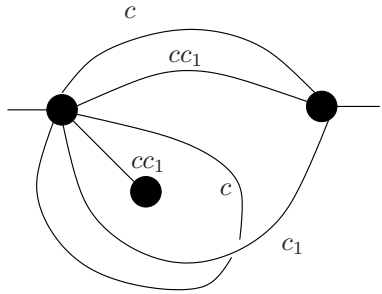}  
\caption{An edge multicolored map with external cilia (the cilia are the half edges in the figure).}
\label{fig:multicoloredmap}
\end{center}
\end{figure}

Let us consider an edge multicolored map $\cM$. For any color $c\in \{1,\dots D\}$, by erasing the edges of colors $\cC$ with $c\notin \cC$, one obtains the \emph{submap} 
$\cM^{c}$ with \emph{external cilia and scars} of $\cM$. 
\begin{definition}
The submap $\cM^c$ \emph{with external cilia and scars} of the edge multicolored map with cilia:
  \[ \cM= \left( \left(\bigsqcup_{ \cC \subset \cD  } \cS^{(\cC)}_{\rm  int} \right) \bigsqcup \cS_{\rm ext} ,\sigma,  \left\{ \alpha^{(\cC)} \Big{|} \cC\subset \cD  \right\} \right) \; ,\] 
  is (see Figure~\ref{fig:facesmulticoloredmap} for an example):
\begin{itemize}
 \item the finite set $\cS^c$ of half edges of color $c$ which is the disjoint union of:
    \begin{itemize}
     \item the \emph{internal half edges} of color $c$: 
        \[ \cS_{\rm  int}^c  = \bigsqcup_{ \cC, \; c\in \cC  } \cS_{\rm  int}^{ ( \cC) } \]
     \item the \emph{external cilia}:
        \[  \cS^c_{\rm ext}   = \cS_{\rm ext}\] 
     \item the \emph{scars}:
        \[ \cS^c_{\rm scars} =  \bigsqcup_{   \cC, \;  c\notin \cC  } \cS_{\rm  int}^{ ( \cC) } \] 
    \end{itemize}
    that is:
   \[ \cS^c =\cS_{\rm  int}^c \bigsqcup \cS^c_{\rm scars}  \bigsqcup \cS^c_{\rm ext} \;. \]
\item the permutation $\sigma$ on $\cS$.
\item the involution $\alpha^c$ on $\cS^c_{\rm int}$ having no fixed points defined by:
   \[
    \alpha^c(h)   = \alpha^{(\cC)}(h) \text{ if } h\in  \cS_{\rm  int}^{(\cC)} \text{ and } c\in \cC\;.
   \]
    The involution $\alpha^c$ is extended to an involution on $\cS$ with fixed points
    by imposing that each scar and each cilium is a fixed point of $\alpha^c$, i.e. $\alpha^c(h) = h$ for $h\in \cS^c \setminus \cS^c_{\rm int}$.
\end{itemize} 
\end{definition}
\begin{figure}[ht]
\begin{center}
\psfrag{c}{$c$}
\psfrag{c1}{$c_1$}
\psfrag{cc1}{$cc_1$}
\psfrag{c2n}{$c_2\neq c,c_1$}
 \includegraphics[width=8cm]{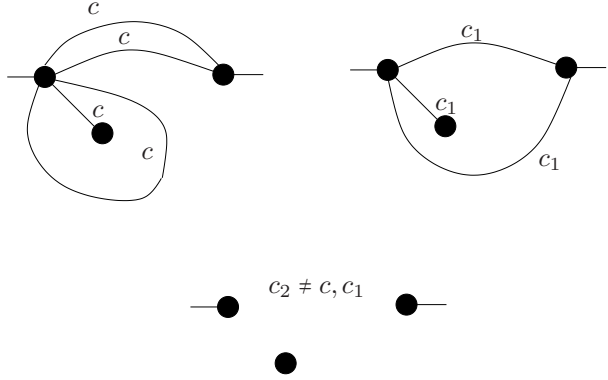}  
\caption{Submaps of color $c$, $c_1$ and respectively $c_2\neq c, c_1$ of the map in Figure~\ref{fig:multicoloredmap} where we did not represent the scars 
in the various submaps.}
\label{fig:facesmulticoloredmap}
\end{center}
\end{figure}

The map $\cM^{c}$ can be disconnected, and some of its connected components can consist in an isolated vertex. 
The corners of an edge multicolored map $\cM$ and its submaps $\cM^c$ are still defined as the couples $\big(h,\sigma(h)\big)$, for all the 
half edges $h$.

\begin{definition}
The \emph{faces} of an edge multicolored map with cilia and scars $\cM^c$ are the cycles of the permutation $\sigma\alpha^c$.
They are divided into
\begin{itemize}
 \item \emph{internal faces.} They are the cycles of $\sigma\alpha^c $ which do not contain any external cilium,
     that is, for every $h$ in the cycle:
        \[
         \alpha^c(h) \notin \cS^c_{\rm ext} \; .
        \]
     The internal faces can however contain any number of scars,
        \[
         \alpha^c(h) =h \in  \cS^c_{\rm scars} \; .
        \]
  \item \emph{external faces.} They are the cycles of $\sigma\alpha^c$ containing cilia. As before, 
  they are further subdivided in sequences of half edges separating two consecutive cilia:
\[  h, \sigma \alpha^c (h) ,\dots , (\sigma \alpha^c)^r (h)  \;  {\rm such \; that} \;
\begin{cases}
 \alpha^c(h) \in \cS^c_{\rm ext} \crcr
  \forall s, \; 0<s<r , \;\; \alpha^c (\sigma \alpha^c)^s (h)  \notin \cS^c_{\rm ext} \crcr 
  \alpha^c  (\sigma \alpha^c)^r (h)  \in \cS^c_{\rm ext}
\end{cases}, \]
 called \emph{external strands}.
\end{itemize}
If the map $\cM^{c}$ has a connected
component consisting in an isolated vertex, it counts as a face. 

The \emph{internal faces and the external strands of color $c$} of an edge multicolored map $\cM$
are the internal faces and the external strands of the submap $\cM^{c}$. We denote by $F^c_{\rm int}(\cM)$
the number of internal faces of color $c$ of $\cM$ and by  $F_{\rm int}(\cM)$ the total number of internal 
faces of $\cM$.
\end{definition}

Remark that the faces of an edge multicolored map are colored by a \emph{unique} color. For example, the faces of colors 
$c$, $c_1$ and $c_2\neq c, c_1$ of the map in Figure~\ref{fig:multicoloredmap} are the faces of its submaps represented
in Figure~\ref{fig:facesmulticoloredmap}. Also, observe that the scars play no role in the definition of the faces. Their only relevance
is to keep track of the corners of the map, and in particular to ensure that an edge multicolored map $\cM$ and its submaps $\cM^c$ have the same 
corners.

\begin{definition}\label{def:boundary}
 The \emph{boundary graph} $\partial \cM$ of an edge multicolored map with cilia $\cM$ is the edge $D$-colored graph 
 obtained by associating a black and a white external vertex to each of its cilia,
 and connecting the white vertex associated to the cilium $h$ with the black vertex associated to the cilium $h'$ by and edge of color $c$
 if there exists an external strand of color $c$ of the map $\cM$ (i.e. of its submap $\cM^c$) going from $h$ to $h'$.
\end{definition}
 
An edge multi colored map can be represented either, as we have done so far, as a map with colored edges or as a \emph{multi ribbon graph}.  
\begin{figure}[ht]
\begin{center}
\psfrag{c}{$c$}
\psfrag{c1}{$c_1$}
\psfrag{cc1}{$cc_1$}
 \includegraphics[width=5cm]{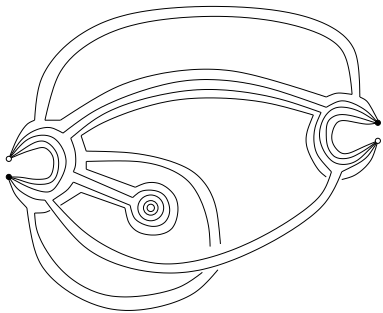}  
\caption{Multi ribbon graph representation of an edge multicolored map with external cilia.}
\label{fig:strandmulticoloredmap}
\end{center}
\end{figure}
This representation is obtained as follows (see Figure~\ref{fig:strandmulticoloredmap} for the representation of the map in
Figure~\ref{fig:multicoloredmap}):
\begin{itemize}
 \item we represent every internal vertex by $D$ concentric strands colored $1$ to $D$. By convention we orient the strands clockwise.
 \item we represent every ciliated vertex as $D$ strands connecting a white and a black external vertex. By convention we orient the strands
 from the white to the black vertex (one can further collapse the white and the black vertex into a cilium).
 \item for any edge of colors $\cC$ we connect the strands with colors in the set $\cC$ on its two end vertices by a ribbon. The edges of the 
 graph are thus made of $|\cC|$ parallel ribbons.
\end{itemize}

The advantage of this second representation is that the faces and the external strands are easily read off: the internal faces are the 
closed strands, while the external strands are the strands connecting a white and a black external vertex.
The boundary graph of $\cM$ is the graph obtained by erasing the internal faces of $\cM$ in this representation.
 
It is important to note that the maps we obtain in the next section \emph{never} have more than one cilium per ciliated vertex.

\section{Feynman rules}\label{sec:Feynrules}
 
 We will now discuss the Feynman rules for the effective theory defined by the Eq. \eqref{eq:smecher} and \eqref{eq:smecherie}.

\subsection{The effective vertices}

The generating function $  Z^{(4)}(J,\bar J)$ can be evaluated perturbatively.
Using for instance the differential operator form for the Gaussian integral, Eq.~\eqref{eq:smecher}, and expanding the exponential we obtain:
\begin{align*}
   Z^{(4)}(J,\bar J) = & 
   \frac{e^{ - N^D \left(  \frac{D a^2 }{2}    +    \ln (1 - g Da) \right) } }{\sqrt{  (1-b^2)^{DN^2-D}   (1-Db^2)  }  } \crcr
   & \times \sum_{n_{\rm  int} , n_{\rm ext} \ge 0}  \frac{1}{n_{\rm int}! n_{\rm ext}!} 
  \Bigg\{ e^{ \frac{1}{2} \Braket { \frac{\partial}{\partial \hat M} | \tilde {\bf O}^{-1} \frac{\partial}{\partial \hat M}  } }
 \;  \Big[  Q(M) \Big]^{n_{\rm int}}  \left(
      \frac{1}{N^{D-1} } 
     \frac{1}{ (1 - g Da )}  J \left(  \frac{1}{ \mathbf{1}^{\cD}  - b M   } \right)\bar J 
   \right)^{n_{\rm ext} }
  \Bigg\}_{M^{c}=0}   \; .
\end{align*}

As usual such an expression is evaluated in terms of Feynman graphs. The graphs are obtained as follows.
We represent the terms $  Q(M)  $ as internal vertices and the $J \bar J$ terms as external vertices. On the external vertices we add a cilium signaling the presence of the external sources $J \bar J$.
The derivatives in the operator $\Braket { \frac{\partial}{\partial \hat M} | \tilde {\bf O}^{-1} \frac{\partial}{\partial \hat M}  } $
act on the vertices, and each derivative creates a half edge. The operator $  \tilde {\bf O}^{-1}$ connects pairs of half edges 
into edges. 

We now evaluate the action of the derivative operators on the various vertices. For the ciliated vertices we have:
\begin{align*}
 \frac{\partial}{\partial M^c_{a^cb^c}} \left[  \frac{1}{ \mathbf{1}^{\cD}  - b M   }  \right]_{q^{\cD}p^{\cD}}
& =  b \; \sum_{   n^{\cD } m^{\cD } } \left(
   \left[  \frac{1}{ \mathbf{1}^{\cD}  - b M   }  \right]_{q^{\cD}   n^{\cD}  } 
    \left[  \frac{1}{ \mathbf{1}^{\cD}  - b M   }  \right]_{  m^{\cD} p^{\cD} }
  \right) \;  \left( \delta_{  n^{\cD\setminus \{c\}} m^{\cD \setminus \{c\} } } \delta_{  n^{c}   a^{c} } 
    \delta_{m^{c} b^{c}} \right) \; .
\end{align*}

For the internal, non ciliated vertices, we need to estimate the derivatives of the subtracted vertex function:
\[
 Q(M) = \Tr_{\cD} \left[ \ln  \frac{1}{ \mathbf{1}^{\cD}  - b M   }  \right] - \Tr_{\cD} \left[ aM \right] - \frac{1}{2}\Tr_{\cD} \left[ (bM )^2\right] \;.
\]
Clearly, starting with three derivatives, the last two terms do not contribute. For the lower orders we use:
\begin{align*}
 \frac{\partial}{\partial M^{c}_{  a^{c}  b^{c} } } \Tr_{\cD} \left[ aM \right] & = 
  \frac{\partial}{\partial M^{c}_{  a^{c}  b^{c} } }  \left( a \sum_{c=1}^D \Tr_{\cD}[M^c \otimes {\bf 1}^{\cD \setminus \{c\}}] \right) = 
  a \sum_{   p^{\cD } q^{\cD } }
   \left[   \mathbf{1}^{\cD}   \right]_{ q^{\cD}    p^{\cD }   } 
   \left( \delta_{ p^{\cD\setminus \{c\} } q^{\cD \setminus \{c\} } } \delta_{  p^{c}   a^{c} } 
    \delta_{q^{c} b^{c}} \right) \; , \crcr
 \frac{\partial}{\partial M^{c}_{  a^{c}  b^{c} } } \frac{1}{2}\Tr_{\cD} \left[ (bM )^2\right] &=
 \frac{b^2}{2}  \frac{\partial}{\partial M^{c}_{  a^{c}  b^{c} } }   \Bigg(  \sum_{c=1}^D \Tr_{\cD}[ (M^c)^2 \otimes {\bf 1}^{\cD \setminus \{c\} } ]  + 
 \sum_{\genfrac{}{}{0pt}{}{c ,c'=1}{c\neq c'}}^D\Tr_{\cD}[ M^c \otimes M^{c'} \otimes {\bf 1}^{\cD \setminus \{c\} \setminus \{c'\} }   
 ]  \Bigg)  \crcr
 & =   b^2  \sum_{   p^{\cD } q^{\cD } }
   M_{ q^{\cD}    p^{\cD }   } 
   \left( \delta_{  p^{\cD\setminus \{c\} } q^{\cD \setminus \{c\} } } \delta_{  p^{c}   a^{c} }
    \delta_{q^{c} b^{c}} \right) \; , \crcr
  \frac{\partial}{\partial M^{c}_{  a^{c}  b^{c} } }  \Tr_{\cD} \left[ \ln  \frac{1}{ \mathbf{1}^{\cD}  - b M   }  \right]   
 & = b \; \sum_{  p^{\cD } q^{\cD } }
   \left[  \frac{1}{ \mathbf{1}^{\cD}  -b M   }  \right]_{ q^{\cD}    p^{\cD }   } 
   \left( \delta_{  p^{\cD\setminus \{c\} } q^{\cD \setminus \{c\} } } \delta_{  p^{c}   a^{c} } 
    \delta_{q^{c} b^{c}} \right) \; .
\end{align*}
It follows that the first derivative of the vertex function is:
\begin{align*}
&  \frac{\partial}{\partial M^{c}_{ a^{c}  b^{c} } } \Big[   Q(M) \Big] =    \sum_{   p^{\cD } q^{\cD } }
   \left[  \frac{b}{ \mathbf{1}^{\cD}  - b M   } -  a \mathbf{1}^{\cD}  - b^2 M\right]_{ q^{\cD}   p^{\cD }   } 
   \left( \delta_{  p^{\cD\setminus \{c\} } q^{\cD \setminus \{c\} } } \delta_{  p^{c}   a^{c} } 
    \delta_{q^{c} b^{c}} \right)   \; ,
\end{align*}
and the second derivative is computed using:
\begin{align*}
& \frac{\partial}{\partial M^{c}_{  a^{c}  b^{c} } }   
   \left[  \frac{b}{ \mathbf{1}^{\cD}  - b M   } -  a \mathbf{1}^{\cD}  - b^2  M\right]_{ q^{\cD}   p^{\cD }   } =
 \crcr
& \quad 
=  b^2 \; \sum_{   n^{\cD } m^{\cD } } \left(
   \left[  \frac{1}{ \mathbf{1}^{\cD}  - b M   }  \right]_{q^{\cD}   n^{\cD}  } 
    \left[  \frac{1}{ \mathbf{1}^{\cD}  - b M   }  \right]_{  m^{\cD} p^{\cD} }
    - [{\bf 1}^{\cD}]_{q^{\cD}   n^{\cD}  } [{\bf 1}^{\cD} ]_{ m^{\cD} p^{\cD}   } 
  \right) \;  \left( \delta_{  n^{\cD\setminus \{c\}} m^{\cD \setminus \{c\} } } \delta_{  n^{c}   a^{c} } 
    \delta_{m^{c} b^{c}} \right) \; .
\end{align*}

\subsection{Broken and unbroken edges}

When evaluating the Gaussian integral in Eq.~\eqref{eq:smecher}, the two terms in the differential  operator 
$ \Braket { \frac{\partial}{\partial \hat M} | \tilde {\bf O}^{-1} \frac{\partial}{\partial \hat M}  } $ correspond to two kinds of edges:
\begin{itemize}
 \item \emph{unbroken edges} colored by a color $c$ coming from the term:
 \[ \Tr_c\left[ \frac{\partial}{\partial M^c} \frac{\partial}{\partial M^c} \right] \; .\]
 \item \emph{broken edges} coming from the term:
  \[  \left( \frac{1}{D} \sum_{c=1}^D\Tr_c\left[  \frac{\partial}{\partial M^c} \right]   \right)
 \left( \frac{1}{D} \sum_{c=1}^D\Tr_c\left[  \frac{\partial}{\partial M^c} \right]   \right) \; .\]
 The broken edges have no color.
\end{itemize}

We draw the broken edges as dashed edges and the unbroken edges as solid edges having a color.
As usual, as the fields are matrices, the order of the half edges around a vertex is relevant and the Feynman graphs are in fact combinatorial maps
(a.k.a. ribbon graphs \cite{DiFrancesco:1993nw}). 
An example of a map with broken and unbroken edges is presented in Figure~\ref{fig:brokenunbroken}.

\begin{figure}[htb]
\begin{center}
\psfrag{c}{$c$}
\psfrag{c1}{$c_1$}
\includegraphics[height=4cm]{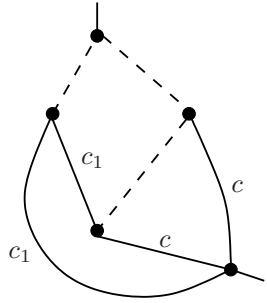}  
\caption{Map with cilia, broken and unbroken colored edges.}\label{fig:brokenunbroken}
\end{center}
\end{figure}

In the perturbative Feynman expansion, the generating function writes as a sum over maps $\cM$ with ciliated vertices whose 
edges fall in two categories: the unbroken edges $\cE_{u}(\cM)$, which have a color $c$, 
and the broken edges, $\cE_{b}(\cM)$. Let us denote $\sigma_{\cM}$ the successor permutation of the map $\cM$, and let us denote the corner
$\big( h,\sigma_{\cM}(h) \big)$ by $[h]$. Estimating the derivatives acting on the vertices, we obtain for every corner of the map an associated matrix element of an operator:
\begin{itemize}
 \item for any corner $[h]$ of a ciliated vertex or of a non ciliated vertex of coordination at least three we
 obtain:
 \[
  b \left[  \frac{1}{ \mathbf{1}^{\cD}  - b M   }  \right]_{ a^{\cD}_{[h]}    b^{\cD }_{[h]}   } \;,
 \]
 \item for a corner $[h]$ of a non ciliated vertex of coordination one we obtain:
 \[
  b  \left[  \frac{1}{ \mathbf{1}^{\cD}  -b M   } -  \frac{a}{b}\mathbf{1}^{\cD}  - b M\right]_{a^{\cD}_{[h]}  b^{\cD }_{[h]}   } \;,
 \]
 \item for the couple of corners $[h]$ and $[\sigma_{\cM}(h)]$ belonging to a non ciliated vertex of coordination two 
 we obtain:
 \begin{align*}
&  b^2 \Big(  \left[  \frac{1}{ \mathbf{1}^{\cD}  - b M   }  \right]_{a^{\cD}_{[h]}   b^{\cD}_{[h]}  } 
    \left[  \frac{1}{ \mathbf{1}^{\cD}  - b M   }  \right]_{  a_{[\sigma(h)]}^{\cD} b_{[\sigma(h)]}^{\cD} } - [{\bf 1}^{\cD}]_{a^{\cD}_{[h]}   b^{\cD}_{[h]}  }    [{\bf 1}^{\cD} ]_{  a_{[\sigma(h)]}^{\cD} b_{[\sigma(h)]}^{\cD} } \Big)  \; .
 \end{align*}
\end{itemize}

The edges contract the indices of the operators incident at their ends in a specific pattern, according to the nature of the edge (unbroken or 
broken). Extracting the logarithm we obtain a sum over connected maps.
We denote $  n_{\rm int}(\cM)$, $n_{\rm ext}(\cM) $, $E_u(\cM)$ and $E_{b}(\cM)$ the numbers of internal vertices, external vertices, 
unbroken edges and broken edges of $\cM$. Recall that $\cS(\cM)$ denotes the set of all the halfedges of $\cM$, and $\cS_{\rm ext}(\cM)$ the set of cilia of $\cM$.
With thees notations, the perturbative expansion of the generating function of the cumulants writes as:
\begin{align*}
    &   \ln Z^{(4)}(J,\bar J)  = - N^D \left(  \frac{D a^2 }{2}    +    \ln (1 - g Da) \right) - \frac{D(N^2-1)}{2} \ln(1-b^2) -\frac{1}{2} \ln(1-Db^2) + \crcr 
   & + 
    \sum_{\cM \; {\rm connected} } \frac{ b^{2E_u(\cM) +2 E_{b}(\cM) } } { n_{\rm int}(\cM)! n_{\rm ext}(\cM)! }  \sum_{a  b} \Bigg{[}
     \left( \prod_{ h \in \cS(\cM) }   \mathfrak{R}(M)_{a^{\cD}_{[h]}  b^{\cD }_{[h]}   }\right)  \left(  \prod_{h\in \cS_{\rm ext}(\cM)} 
    \frac{  J_{ a^{\cD}_{ [h] } } \bar J_{  b^{\cD}_{ [ \sigma_{\cM}^{-1}(h) ] } } }{N^{D-1} (1- g Da ) } \right)
  \crcr
   &  \;\;  \times \Big[ \prod_{e^{c}= (h,h') \in \cE_u(\cM) } \frac{1}{N^{D-1}(1-b^2)}
    \left(  \delta_{     b^{c}_{ [  \sigma_{\cM}^{-1}(h) ] } a^{c}_{[h'] } } 
     \delta_{   b^{c}_{   [ \sigma_{\cM}^{-1}(h')] } a^{c}_{[h] } } \right)   \left( \delta_{ b^{\cD\setminus \{c\} }_{ [ \sigma_{\cM}^{-1}(h)] } a^{\cD\setminus \{c\} }_{ [h] }}
    \delta_{   b^{\cD\setminus \{c\} }_{[  \sigma_{\cM}^{-1}(h') ] }  a^{\cD\setminus \{c\} }_{ [h']  }  } \right) \Big] 
    \crcr
   &  \;\;  \times \Big[ \prod_{e= (h,h') \in \cE_{b }(\cM) } \frac{b^2 D (D-1)} {N^{D}(1-b^2)(1-Db^2)}    \left( \delta_{ b^{\cD }_{ [ \sigma_{\cM}^{-1}(h)] } a^{\cD   }_{ [h] }}
    \delta_{   b^{\cD }_{[  \sigma_{\cM}^{-1}(h') ] }  a^{\cD }_{ [h']  }  } \right) \Big] 
    \Bigg{]}_{M=0} \; ,
\end{align*}
where we have used the notation:
\begin{itemize}
 \item for any corner $[h]$ of a ciliated vertex or of a non ciliated vertex of coordination at least three: 
 \[
  \mathfrak{R}(M)_{a^{\cD}_{[h]}  b^{\cD }_{[h]}   }  \equiv \left[  \frac{1}{ \mathbf{1}^{\cD}  - b M   }  \right]_{ a^{\cD}_{[h]}    b^{\cD }_{[h]}   } \;,
 \]
 \item for a corner $[h]$ of a non ciliated vertex of coordination one:
 \[
 \mathfrak{R}(M)_{a^{\cD}_{[h]}  b^{\cD }_{[h]}   }  \equiv
 \left[  \frac{1}{ \mathbf{1}^{\cD}  - b M   } -  \frac{a}{b}\mathbf{1}^{\cD}  - b M\right]_{a^{\cD}_{[h]}  b^{\cD }_{[h]}   } \;,
 \]
 \item for the couple of corners $[h]$ and $[\sigma_{\cM}(h)]$ belonging to a non ciliated vertex of coordination two:
 \begin{align*}
 & \mathfrak{R}(M)_{a^{\cD}_{[h]}  b^{\cD }_{[h]}   }  \mathfrak{R}(M)_{a^{\cD}_{[\sigma(h)]}  b^{\cD }_{[\sigma(h)]}   } 
 \equiv \left[  \frac{1}{ \mathbf{1}^{\cD}  - b M   }  \right]_{a^{\cD}_{[h]}   b^{\cD}_{[h]}  } 
    \left[  \frac{1}{ \mathbf{1}^{\cD}  - b M   }  \right]_{  a_{[\sigma(h)]}^{\cD} b_{[\sigma(h)]}^{\cD} }
    - [{\bf 1}^{\cD}]_{a^{\cD}_{[h]}   b^{\cD}_{[h]}  }    [{\bf 1}^{\cD} ]_{  a_{[\sigma(h)]}^{\cD} b_{[\sigma(h)]}^{\cD} } \; .
 \end{align*}
\end{itemize}

\section{Translating to the melonic vacuum}\label{sec:translatevacuum}

If we translate to the melonic vacuum, several significant simplifications can be made. Using $b(a_0)=a_0$  and taking into account that $M$ is set to zero we obtain:
\begin{itemize}
 \item for any corner $[h]$ of a ciliated vertex or of a non ciliated vertex of coordination at least three: 
 \[
  \mathfrak{R}(0)_{a^{\cD}_{[h]}  b^{\cD }_{[h]}   } =  \mathbf{1}^{\cD}_{ a^{\cD}_{[h]}    b^{\cD }_{[h]}   } \;,
 \]
 \item for a corner $[h]$ of a non ciliated vertex of coordination one:
 \[
 \mathfrak{R}(0)_{a^{\cD}_{[h]}  b^{\cD }_{[h]}   } = \frac{b(a_0)-a_0}{b(a_0) } \mathbf{1}^{\cD}_{ a^{\cD}_{[h]}    b^{\cD }_{[h]}   }=0\;,
 \]
 \item for the couple of corners $[h]$ and $[\sigma_{\cM}(h)]$ belonging to a non ciliated vertex of coordination two:
 \begin{align*}
 & \mathfrak{R}(0)_{a^{\cD}_{[h]}  b^{\cD }_{[h]}   }  \mathfrak{R}(0)_{a^{\cD}_{[\sigma(h)]}  b^{\cD }_{[\sigma(h)]}   } \crcr
 & \quad =
 [{\bf 1}^{\cD}]_{a^{\cD}_{[h]}   b^{\cD}_{[h]}  }    [{\bf 1}^{\cD} ]_{  a_{[\sigma(h)]}^{\cD} b_{[\sigma(h)]}^{\cD} } 
 - [{\bf 1}^{\cD}]_{a^{\cD}_{[h]}   b^{\cD}_{[h]}  }    [{\bf 1}^{\cD} ]_{  a_{[\sigma(h)]}^{\cD} b_{[\sigma(h)]}^{\cD} }
 =0 \; .
\end{align*}
\end{itemize}

Thus only maps whose internal vertices have coordination at least three (which we denote $\cM^3$) survive. 
Moreover, a map $\cM^3$ has unbroken colored edges and broken edges. Erasing all the broken edges leads to a 
map whose edges are colored by a color $c=\{1,\dots D\}$. We denote this map by $\cM^3_u$ (as it is formed only by the unbroken edges). 
Remark that $\cM^3_u$ can be disconnected, even though $\cM^3$ is connected (see Figure~\ref{fig:m3u}). 
\begin{figure}[htb]
\begin{center}
\includegraphics[height=4cm]{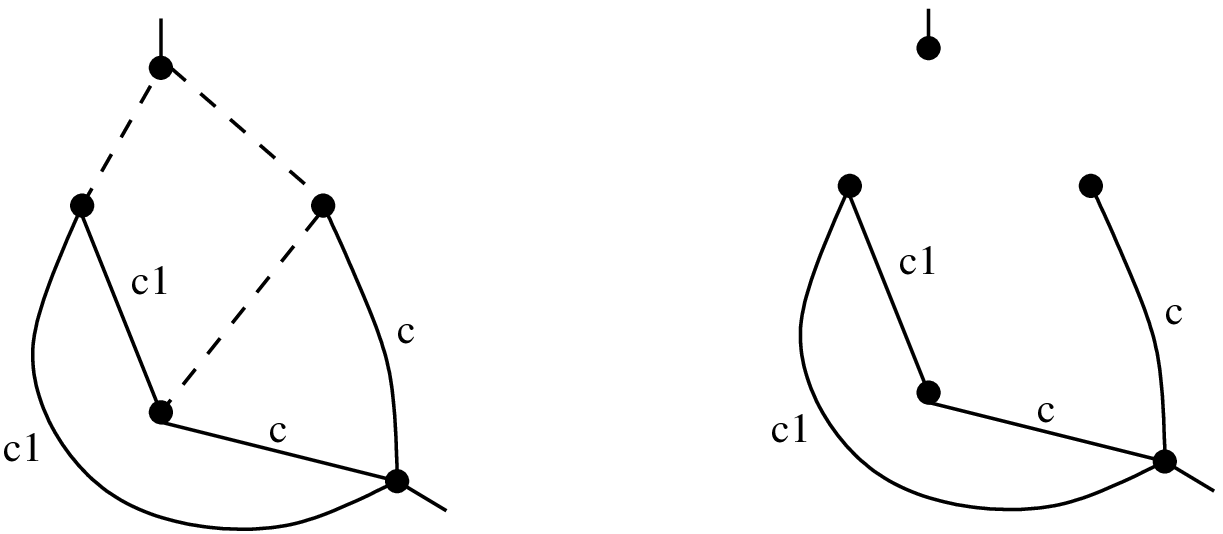}  
\caption{The maps $\cM^3$ and $\cM^3_u$.}\label{fig:m3u}
\end{center}
\end{figure}

As $\cM^{3}_u$ is a usual edge colored map it has a well defined notion of boundary graph $\partial \cM^3_u$ 
and (internal and external) faces. The indices are identified along the faces of $\cM^3_u$.
As the numbers of unbroken edges of $\cM^3$ and $\cM^3_u$ are equal, and the number of ciliated vertices of $\cM$
is the number of white vertices of $\partial \cM^3_u$ we obtain:
\begin{align*}
   \ln Z^{(4)}(J,\bar J)   =  &
 - N^D \left(  \frac{D a^2 }{2}    +    \ln (1 - g Da) \right) - \frac{D(N^2-1)}{2} \ln(1-b^2) -\frac{1}{2} \ln(1-Db^2) + \crcr 
& + \sum_{\cM^3 \; {\rm connected} } \frac{ \Tr_{\partial \cM^3_u} ( J,\bar J) }{  k(\partial \cM^3_u)!} 
\frac{1}{ n_{\rm int}(\cM^3)!}
\crcr
& \qquad \qquad \times \frac{ a_0^{2 E_u(\cM^3) +2 E_b(\cM^3)    } }{ (1- g Da_0 )^{k(\partial \cM^3_u )}  } 
\left[ \frac{1}{1-a_0^2}\right]^{E_u(\cM^3 ) }
\left[ \frac{a_0^2 D (D-1)} { (1-a_0^2)(1-Da_0^2)}  \right]^{E_{b }(\cM^3)} \crcr
&\qquad \qquad \times  N^{-D E_{b }(\cM^3) -(D-1)E_u(\cM^3 )    - (D-1) k(\partial\cM^3_u) + F_{\rm int}(\cM^3_u) } \; ,
\end{align*}
where $k(\partial \cM^3_u ) $ denotes the number of white vertices of the boundary graph $ \partial \cM^3_u$ (see definition \ref{def:boundary}) and 
\[
  \Tr_{\partial \cM^3_u} ( J,\bar J) = \prod_{i=1}^{k(\partial \cM^3_u  )} \bar J_{\bar a^{\cD}_i} J_{a^{\cD}_i}
  \prod_{e^c= (v_i,\bar v_j)\in \cE(\cM^3_u)} \delta_{a^c_i \bar a^c_j} \; .
\]
The rescaled cumulants are computed by taking derivatives and dividing by the appropriate power of $N$:
\begin{align}\label{eq:cumulshift1/N}
 K(\cB, \mu^{(4)}) & =  \sum_{ \genfrac{}{}{0pt}{}{\cM^3 \; {\rm connected}}{\partial \cM^3_u = \cB} } 
\frac{1}{ n_{\rm int}(\cM^3)!}    \frac{ a_0^{2E_u(\cM^3 ) + 4 E_{b}(\cM^3)  } }{ (1- g  Da_0 )^{k(\partial\cM^3_u)}  } 
 \frac{   [D(D-1)]^{  E_{b }(\cM^3)   }  }{ (1-a_0^2)^{  E_u(\cM^3 ) + E_{b }(\cM^3) } (1-Da_0^2)^{E_{b }(\cM^3)   }  }  \crcr
&\qquad \qquad  \times  N^{-D E_{b }(\cM^3) -(D-1)E_u(\cM^3) -D + (D-1) k(\partial\cM^3_u) +C(\partial \cM^3_u)   + F_{\rm int}(\cM^3_u) } \; . 
\end{align}

\subsection{The $1/N$ expansion for the fluctuation field}

In this section we show that the cumulants admit a $1/N$ expansion.
We label $\cM^3_{u;(\rho)}$ the connected components of $\cM^3_u$ and we denote $C(\cM^3_u)$ their number. 

\begin{definition}
For each connected component $\cM^3_{u;(\rho)}$ we define its \emph{deficit} $\eta (  \cM^3_{u;(\rho)} )   $ as:
\begin{align*}
&  \eta  (\cM^3_{u;(\rho)} )   = 2 +  E_u(\cM^3_{u;(\rho)} ) +(D-2)  V(\cM^3_{u;(\rho)}) -\crcr
& \qquad \qquad \qquad  - (D-1) k(\partial\cM^3_{u;(\rho)} ) - C(\partial \cM^3_{u;(\rho)} ) - F_{\rm int}(\cM^3_{u;(\rho)})  \;. 
\end{align*}
Furthermore, we define the \emph{excess of broken edges} of $\cM^3$, $L_b(\cM^3)$, and the \emph{excess of unbroken edges} of $\cM^3_{u;(\rho)} $,
$L_u( \cM^3_{u;(\rho)} )$, as:
\[ L_b(\cM^3) \equiv E_{b }(\cM^3)  - C(\cM_u^3)  +1 \; , \;\;  L_u( \cM^3_{u;(\rho)} ) \equiv E_u(\cM^3_{u;(\rho)} ) - V(\cM^3_{u;(\rho)}) + 1 \;. \]
\end{definition}

\begin{theorem}[$1/N$ expansion]\label{thm:1/Nintfield}
 The rescaled cumulants of $\mu^{(4)}$ can be expressed as the formal series:
\begin{align}\label{eq:cumulshift1/Nbun}
    K(\cB, \mu^{(4)}) = & \sum_{ \genfrac{}{}{0pt}{}{\cM^3 \; {\rm connected}}{\partial \cM^3_u = \cB} } 
\frac{1}{ n_{\rm int}(\cM^3)!}  \frac{ a_0^{2E_u(\cM^3 ) + 4 E_{b}(\cM^3)  } }{ (1- g Da_0 )^{k(\partial\cM^3_u)}  } 
 \frac{   [D(D-1)]^{  E_{b }(\cM^3)   }  }{ (1-a_0^2)^{  E_u(\cM^3 ) + E_{b }(\cM^3) }   (1-Da_0^2)^{E_{b }(\cM^3)   }  }
\crcr
&  \qquad \qquad   \times  
N^{
 -  D L_b(\cM^3)  - (D-2) \sum_{\rho=1}^{C(\cM^3_u)}   L_u( \cM^3_{u;(\rho)} ) - \sum_{\rho=1}^{C(\cM^3_u)} \eta (\cM^3_{u;(\rho)} )  
} \; . 
\end{align}
 
The right hand side of Eq.~\eqref{eq:cumulshift1/Nbun} is a series in $1/N$.
More importantly, there are only a finite number of maps $\cM^3$ contributing to any fixed order in $1/N$.
\end{theorem}
\begin{proof}
A straightforward computation leads from Eq.~\eqref{eq:cumulshift1/N} to Eq.~\eqref{eq:cumulshift1/Nbun}
taking into account that the boundary graph and the internal faces of $\cM_u^3$ are distributed among the connected 
components $\cM^3_{u;(\rho)}$ and that the numbers of vertices, ciliated vertices and 
unbroken edges of $\cM^3$ and $\cM^3_u$ are equal.

As $\cM^3$ is connected, all the connected components of $\cM^3_u$ must be connected in between them by 
broken edges, hence $L_{b}(\cM^3)\ge 0$. 
Similarly, as $\cM^3_{u;(\rho)}$ is connected, $L_u( \cM^3_{u;(\rho)} )   \ge 0$. 

In appendix \ref{app:proof} we prove the following lemma:
\begin{lemma}\label{lem:facesbound}
 The number of internal faces of a connected edge colored map with $k(\partial\cM)$ cilia is bounded by:
\begin{align*}
     F_{\rm int} (\cM) & \le 1 - (D-1) k(\partial\cM) -C(\partial\cM) + (D-1)V(\cM)   + [E(\cM) - V(\cM)+1] \; .
 \end{align*}
 \end{lemma}
 
As a consequence of this lemma, for any connected component $\cM^3_{u;(\rho)}$, its deficit $\eta  (\cM^3_{u;(\rho)} )$ is non negative, 
$\eta  (\cM^3_{u;(\rho)} ) \ge 0 $. It follows that the right hand side of Eq.~\eqref{eq:cumulshift1/Nbun} is a series in $1/N$.
 
In order to prove the second part of the theorem, we note that, as $D > D-2$, the scaling in $N$ is bounded by:
\begin{align*}
 & -(D-2) \Big[ L_b(\cM^3)  +  \sum_{\rho=1}^{C(\cM^3_u)}  L_u( \cM^3_{u;(\rho)} )  \Big] = -(D-2) \big[  E_{b }(\cM^3) + E_u(\cM^3 ) - V(\cM^3) +1\big] \; .
\end{align*}

As the ciliated vertices in the map $\cM^3$ have coordination at least one, while the non ciliated vertices have coordination at least three,
we have:
 \begin{align*}
  2 \big[E_{b }(\cM^3)   + E_u(\cM^3) \big]  \ge 3 \left[ V(\cM^3) - k(\partial \cM^3_u) \right] + k(\partial\cM^3_u) \Rightarrow 
 V(\cM^3) 
  \le \frac{ 2 }{3} \big[ E_{b }(\cM^3)  + E_u(\cM^3)  +   k(\partial\cM^3_u) \big] \; ,
 \end{align*}
hence the scaling with $N$ of a term on the right hand side of Eq.~\eqref{eq:cumulshift1/N} is bounded by:
\[
    - (D-2) \left[\frac{  E_{b }(\cM^3)   + E_u(\cM^3)  }{3}- \frac{2 k(\partial\cM^3_u)}{3}+1 \right] \;.
\]

The number of external vertices of $\cM^3_u$ is fixed, as it equals the number of the white vertices of the invariant $\cB$ whose cumulant we
evaluate. It follows that at any order in $1/N$ only connected 
maps with at most a finite number of edges contribute, and the second part of the theorem follows.  

\end{proof}
 
 Setting the external sources to zero we obtain the free energy of the quartic melonic model:
 \begin{align*}
   \ln Z^{(4)} =  &
 - N^D \left(  \frac{D a^2 }{2}    +    \ln (1 - g Da) \right) - \frac{D(N^2-1)}{2} \ln(1-b^2) -\frac{1}{2} \ln(1-Db^2) + \crcr 
& + N^D \sum_{ \genfrac{}{}{0pt}{}{ \cM^3 \; {\rm connected} }{ \partial  \cM^3 = \emptyset } } 
\frac{1}{ n_{\rm int}(\cM^3)!}
\crcr
& \qquad \qquad \times  a_0^{2 E_u(\cM^3) +2 E_b(\cM^3)    }  
\left[ \frac{1}{1-a_0^2}\right]^{E_u(\cM^3 ) }
\left[ \frac{a_0^2 D (D-1)} { (1-a_0^2)(1-Da_0^2)}  \right]^{E_{b }(\cM^3)} \crcr
&\qquad \qquad \times  N^{
 -  D L_b(\cM^3)  - (D-2) \sum_{\rho=1}^{C(\cM^3_u)}   L_u( \cM^3_{u;(\rho)} ) - \sum_{\rho=1}^{C(\cM^3_u)} \eta (\cM^3_{u;(\rho)} )  
} \; .
\end{align*}
The sum over $\cM^3$ does not contribute at leading oder in $N$. Indeed, a term in the sum with  $  L_b(\cM^3) =  L_u( \cM^3_{u;(\rho)} ) = 0$ would be a tree, hence would have univalent vertices, 
which is impossible as all the vertices $\cM^3$ are at least three valent.

\appendix

\section{Proof of lemma \ref{lem:facesbound}}\label{app:proof}

In this appendix we show that the number of internal faces of a connected edge colored map $\cM$ cilia is bounded by:
\begin{align*}
     F_{\rm int} (\cM) & \le 1 - (D-1) k(\partial \cM) -C(\partial\cM) + (D-1)V(\cM)   + [E(\cM) - V(\cM)+1] \; .
 \end{align*}

 \begin{proof}
In order to establish this bound, we first introduce a new colored map $\cM^*$ associated to $\cM$ which will allow us to keep track of the boundary $\partial\cM$ while modifying $\cM$. 
Starting from $\cM$, for each cilium $h$, 
we introduce $2D$ half edges on the corner bearing the cilium $h$, one at its left and one at its right for each color $c\in \cD$ (that is, for each color $c$,
one half edge precedes $h$ and the other succeeds $h$ when turning around the vertex). We then connect these new half edges into dashed edges
following the edges of $\partial\cM$: if the external strand of color $c$ starting at the cilium $h$ ends at the cilium $h'$, we connect the half edge of color $c$ following $h$ with the half edge of color $c$
preceding $h'$. This construction is represented in figure \ref{fig:tauedges}.

 \begin{figure}[htb]
\begin{center}
\includegraphics[height=1.8cm]{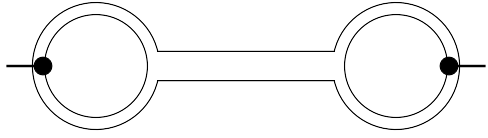}  
\caption{A map $\cM$ with two cilia. For simplicity, only the colors 1 and 2 have been represented.}\label{fig:simpleM}
\end{center}
\end{figure}

Each of the external strands of $\cM$ corresponds to an edge in the boundary graph $\partial \cM$, hence $\cM$ has exactly $Dk(\partial\cM)$ external strands. 
By construction $\cM^*$ has a new edge for each external strand of $\cM$ which closes the external strand of $\cM$ into an internal face of $\cM^*$.
It follows that:
\[
F_{\rm int} (\cM^*) = F_{\rm int} (\cM)+ Dk(\partial\cM)\ .
\]

On the other hand, as the new dashed edges exactly follow the external strands of $\cM$, it ensues that $\partial \cM^* = \partial \cM$.

\begin{figure}[htb]
\begin{center}
\includegraphics[height=3cm]{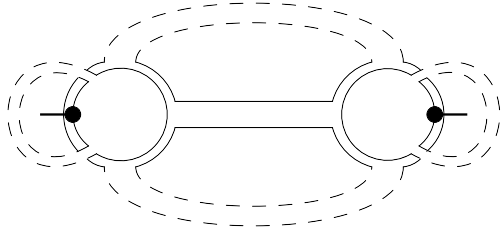}  
\caption{The map $\cM^*$ obtained from the map $\cM$ in Fig. \ref{fig:simpleM} by adding the dashed edges.}\label{fig:tauedges}
\end{center}
\end{figure}

The reason to introduce $\cM^*$ is the following. Below we will delete solid edges of $\cM^*$ (that is edges which belonged to $\cM$) 
and track the change in the number of internal faces of $\cM^*$ under these deletions. The crucial point is that these deletions 
modify the structure of the internal faces of $\cM^*$ but \emph{do not modify} its external strands, hence the boundary graph $\partial\cM^*$ will remain unchanged.
The introduction of $\cM^*$ allows us to modify a map while keeping its boundary graph unchanged.

Let $\cT$ be a tree in $\cM$ and let us denote $\cT^*$ the graph made of $\cT$ and all dashed edges of $\cM^*$. We have:
\[\partial\cT^*=\partial \cM^*=\partial\cM \;.\]

The map $\cT^*$ can be obtained from $\cM^*$ by deleting $E(\cM)-V(\cM)+1$ solid edges. As there is only one face going trough any 
of the solid edges, the deletion of such an edge in $\cM^*$ can either increase of decrease the number of faces by one, hence:
\begin{align*}
F_{\rm int} (\cM)+ Dk(\partial \cM)= F_{\rm int}(\cM^*)\ \leq\ F_{\rm int}(\cT^*) \,+\, E(\cM)-V(\cM)+1 \ .
\end{align*}

We call a \emph{leaf} of $\cT^*$ a vertex which becomes univalent when one erases all the dashed edges. A moments reflexion
reveals that the leafs of $\cT^*$ are the univalent vertices of $\cT$.
In order to complete the proof of lemma \ref{lem:facesbound}, it is now enough to show that for any $\cT^*$ we have: 
\begin{align*}
 F_{\rm int}(\cT^*) + C(\partial\cT^*) \leq 1 + k(\partial\cT^*) + (D-1)V(\cT^*)  \ .
\end{align*}

This bound is obtained by tracking the evolution of $F_{\rm int}(\cT^*) + C(\partial\cT^*)$ under the iterative deletion of the leaves. 
This deletion is somewhat involved, as it can change the boundary graph.

We say that an external strand of $\cT^*$ is \emph{looped at the cilium $h$} if it starts and ends at the cilium $h$. On the boundary graph $\partial \cT^*$, 
this corresponds to an edge connecting the black and white vertices associated to $h$. 
In figure \ref{fig:tauedges} we can identify two looped external strands.
Let us denote $\ell$ a leaf of $\cT^*$ and let us denote $\hat\cT^*$ the tree obtained from $\cT^*$ after deleting $\ell$ as follows.

\paragraph{Deleting a non ciliated leaf.} 
If $\ell$ has no cilium, deleting it simply means deleting the vertex $\ell$ and the edge connecting it to the rest of $\cT^*$. 
If the edge connecting $\ell$ to the rest of $\cT^*$ has color $c$, the deletion of $\ell$ erases all the faces of colors 
$\cD\setminus \{c\}$ running trough $\ell$. The boundary graph is unaltered by this deletion, therefore:
\begin{align*}
 F_{\rm int} (\cT^*)  = F_{\rm int} (\hat\cT^*) + (D-1) \; , & \qquad 
  C(\partial\cT^*) =C(\partial\hat\cT^*) \Rightarrow \crcr
F_{\rm int} (\cT^*)+C(\partial\cT^*) &= F_{\rm int} (\hat\cT^*) + C(\partial\hat\cT^*) + (D-1)\ .
 \end{align*}

\paragraph{Deleting a ciliated leaf.} 
If $\ell$ has a cilium $h$, deleting it consists in :
\begin{itemize}
 \item {\it Step 1.} For all the external strands starting or ending at $h$ that are not looped, we cut the corresponding dashed edges into half edges. 
 We then reconnect the dashed half edges into edges the other way around, respecting the colors. This is represented in figure 
 \ref{fig:deletion1} below. Observe that after performing this step, all the external strands going trough $\ell$ are looped.

 \begin{figure}[htb]
\begin{center}
\includegraphics[height=3.5cm]{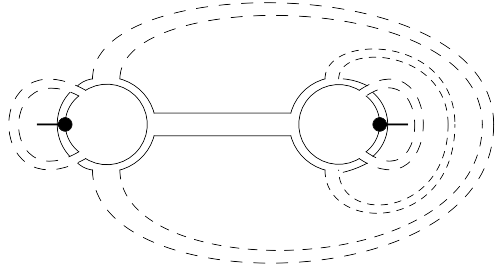}  
\caption{The first step of the deletion of a ciliated leaf : the dashed edges have been reconnected in such way that all the external strands are looped at the cilium $h$.}\label{fig:deletion1}
\end{center}
\end{figure}
 \item{\it Step 2.} The leaf now has only looped external strands and it is connected to the rest of the $\cT^*$ by a solid edge only.
 The cilium $h$ represents a connected component of the boundary graph consisting in a black and a white vertex connected by $D$ edges. 
 We erase the vertex $\ell$, its cilium, $h$, and the edge connecting $\ell$ to the rest of the graph, as in figure \ref{fig:deletion2}
\begin{figure}[htb]
\begin{center}
\includegraphics[height=3.2cm]{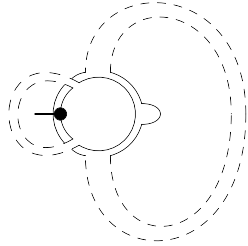}  
\caption{The second step of the deletion of a ciliated leaf.}\label{fig:deletion2}
\end{center}
\end{figure}
 \end{itemize}
 
{\it Boundary graph.} Upon deleting a ciliated leaf, the boundary graph changes: $\partial \hat \cT^*\neq \partial \cT^*$. We have two cases, each with two sub cases:
\begin{itemize}
\item None of the external strands of $\cT^*$ are looped at $h$. One needs to apply Step 1 for all the colors and then Step 2. 
There are two sub cases:
\begin{itemize}
 \item  the black and white vertices associated to $h$ in $\partial \cT^*$ belong to the same connected component of $\partial\cT^*$. Then, in the boundary graph,
 Step 1 creates at least a new connected component, and Step 2 deletes exactly one connected component. Thus:
   \[  C(\partial\cT^*) \leq  C(\partial\hat\cT^*) \; .\]
 \item the black and white vertices associated to $h$ in $\partial \cT^*$ belong to different connected components of $\partial\cT^*$. Then, in the boundary graph,
 Step 1 can not decrease the number of connected components, and Step 2 deletes exactly one connected component. Thus:
   \[  C(\partial\cT^*)  \leq C(\partial\hat\cT^*) +1 \; .\]
 \end{itemize}
\item At least one external strands of $\cT^*$ is looped at $h$. There are two sub cases:
\begin{itemize}
 \item not all the external strands are looped at $h$. The black and white vertices associated to $h$ in $\partial \cT^*$ belong to the same connected component of $\partial\cT^*$.
  One must apply Step 1 at least once, and then step 2. As before, Step 1 creates at least a new connected component, and Step 2 deletes exactly one connected component,
  hence:
  \[
   C(\partial\cT^*) \leq  C(\partial\hat\cT^*) \; .
  \]
\item all the external strands are looped at $h$. The black and white vertices associated to $h$ in $\partial \cT^*$ belong to two different connected components of $\partial\cT^*$,
   and one must apply Step 2 directly. This decrease the number of connect components of the boundary graph by $1$:
  \[
   C(\partial\cT^*)  =  C(\partial\hat\cT^*) +1 \; .
  \]
\end{itemize}
\end{itemize}
 
 {\it Internal faces.} Let us denote $e^c$ the solid edge of color $c$ connecting $\ell$ to the rest of $\cT^*$. We have several cases:
 \begin{itemize}
  \item the external strand of color $c$ is looped at $h$. Then there is only one internal face of color $c$ in $\cT^*$ running through $e^c$, which can not be erased by deleting $\ell$, hence:
 \[ F_{\rm int}^c(\cT^*)  =F_{\rm int}^c(\hat\cT^*) \; . \]
  \item the external strand of color $c$ is not looped at $h$. Then there are either one or two internal faces of color $c$ in $\cT^*$ running trough $e^c$, hence the number of internal 
  faces of color $c$ can not decrease by more than one:
      \[ F_{\rm int}^c(\cT^*)  \leq F_{\rm int}^c(\hat\cT^*) +1 \;. \]
    \item the external strand of color $c'\neq c$ is looped. Then there is only one internal face of color $c'$ through $\ell$, which is erased: 
       \[ F_{\rm int}^{c'}(\cT^*) = F_{\rm int}^{c'}(\hat\cT^*)+1 \; .\]
  \item the external strand of color $c'\neq c$ is not looped. Then there is just one internal face of color $c'$ through $\ell$, which is not erased: 
  \[F_{\rm int}^{c'}(\cT^*) = F_{\rm int}^{c'}(\hat\cT^*) \; .\] 
\end{itemize}

Combining the counting of the connected components of the boundary graphs with the counting of the internal faces we obtain four cases:
\begin{itemize}
 \item no external strand is looped at $h$. Then:
     \[ 
      F_{\rm int} (\cT^*)+C(\partial\cT^*) \le [F_{\rm int} (\hat \cT^*) + 1] + [ C(\partial \hat \cT^*) + 1] \le F_{\rm int} (\hat \cT^*)+C(\partial \hat \cT^*) + D \;. 
     \]
 \item all the external strands are looped at $h$. Then:
      \[
        F_{\rm int} (\cT^*)+C(\partial\cT^*) \le [F_{\rm int} (\hat \cT^*) + D-1] + [ C(\partial \hat \cT^*) + 1 ] \le F_{\rm int} (\hat \cT^*)+C(\partial \hat \cT^*) + D \;. 
      \]
 \item the external strand of color $c$ is looped at $h$, but at least one external strand of color $c'\neq c$ is not looped at $h$. Then:
      \[
        F_{\rm int} (\cT^*)+C(\partial\cT^*) \le [ F_{\rm int} (\hat \cT^*) + D-2 ]+ [ C(\partial \hat \cT^*) ] \le F_{\rm int} (\hat \cT^*)+C(\partial \hat \cT^*) + D \;.
      \]
\item the external strand of color $c$ is not looped at $h$, but some external strands of colors $c'\neq c$ are looped at $h$. Then
       \[
        F_{\rm int} (\cT^*)+C(\partial\cT^*) \le [ F_{\rm int} (\hat \cT^*)  + D ] + [C(\partial \hat \cT^*)] \le F_{\rm int} (\hat \cT^*)+C(\partial \hat \cT^*) + D \;.
      \]
\end{itemize}

In all cases, by deleting a ciliated leaf:
 \begin{equation*}
  F_{\rm int}^c(\cT^*) + C(\partial\cT^*) \leq F_{\rm int}^c(\hat\cT^*) + C(\partial\hat\cT^*) + D \; .
 \end{equation*}
Iterating up to the last vertex, we either end up with a vertex with no cilium or with a ciliated vertex with $D$ looped external strands.
Counting the number of internal faces and connected component of the two possible final graphs we conclude.
 
 \end{proof}
 
\newpage

\bibliography{Refs}{}

   \end{document}